\RequirePackage[l2tabu,orthodox]{nag}
\documentclass
[11pt,letterpaper]
{article}

\newcommand{\Erdos}{Erd\H{o}s\xspace}
\newcommand{\Renyi}{R\'enyi\xspace}
\newcommand{\Lovasz}{Lov\'asz\xspace}

\usepackage[notes=false,later=false,camera=false]{dtrt}
\usepackage[utf8]{inputenc}
\usepackage{etex}
\usepackage{ stmaryrd }
\usepackage{xspace,enumerate}
\usepackage[T1]{fontenc}
\usepackage[full]{textcomp}
\usepackage[american]{babel}
\usepackage{mathtools}

\usepackage{amsthm}
\usepackage{empheq}

   \usepackage{hyperref}
\usepackage[capitalise,nameinlink]{cleveref}
\crefname{lemma}{Lemma}{Lemmas}
\crefname{fact}{Fact}{Facts}
\newcommand{\colorconstraints}{\text{Color Constraints}}
\crefname{colorconstraints}{(color constraints)}{Color Constraints}
\crefformat{colorconstraints}{#2\colorconstraints#3}
\crefname{indsetconstraints}{(indset constraints)}{IndSet Constraints}
\crefformat{indsetconstraints}{#2$\mathsf{IndSet\ Axioms}$#3}
\crefname{theorem}{Theorem}{Theorems}
\crefname{mtheorem}{Theorem}{Theorems}
\crefname{corollary}{Corollary}{Corollaries}
\crefname{claim}{Claim}{Claims}
\crefname{example}{Example}{Examples}
\crefname{algorithm}{Algorithm}{Algorithms}
\crefname{problem}{Problem}{Problems}
\crefname{definition}{Definition}{Definitions}
\usepackage{paralist}
\usepackage{turnstile}
\usepackage{mdframed}
\usepackage{tikz}
\usepackage{caption}
\DeclareCaptionType{Algorithm}
\usepackage{newfloat}
\newtheorem{theorem}{Theorem}[section]
\newtheorem*{theorem*}{Theorem}

\newtheorem{proposition}[theorem]{Proposition}
\newtheorem*{proposition*}{Proposition}
\newtheorem{lemma}[theorem]{Lemma}
\newtheorem*{lemma*}{Lemma}

\newtheorem*{conjecture*}{Conjecture}
\newtheorem{fact}[theorem]{Fact}
\newtheorem*{fact*}{Fact}

\newtheorem*{hypothesis*}{Hypothesis}

\theoremstyle{definition}
\newtheorem{definition}[theorem]{Definition}
\newtheorem*{definition*}{Definition}

\newtheorem{algorithm}[theorem]{Algorithm}

\theoremstyle{remark}
\newtheorem{claim}[theorem]{Claim}
\newtheorem*{claim*}{Claim}
\newtheorem{remark}[theorem]{Remark}
\newtheorem*{remark*}{Remark}

\newtheorem*{observation*}{Observation}

\usepackage[
letterpaper,
top=1.2in,
bottom=1.2in,
left=1in,
right=1in]{geometry}
\usepackage{newpxtext} %
\usepackage{textcomp} %
\usepackage[varg,bigdelims]{newpxmath}
\usepackage[scr=rsfso]{mathalfa}%
\usepackage{bm} %
\let\mathbb\varmathbb
\usepackage{microtype}

\newcommand{\david}[1]{\dtcolornote[David]{green}{#1}}
\newcommand{\pravesh}[1]{\dtcolornote[Pravesh]{red}{#1}}

\allowdisplaybreaks

\newcommand{\set}[1]{\{#1\}}
\newcommand{\Set}[1]{\left\{#1\right\}}

\newcommand{\bignorm}[1]{\big\lVert#1\big\rVert}
\newcommand{\Bignorm}[1]{\Big\lVert#1\Big\rVert}

\newcommand{\R}{{\mathbb R}}
\newcommand{\N}{{\mathbb N}}
\newcommand{\norm}[1]{\lVert #1 \rVert}

\newcommand{\abs}[1]{\lvert #1 \rvert}

\let\epsilon=\varepsilon

\newcommand{\e}{\varepsilon}

\newcommand{\E}{{\mathbb E}}
\newcommand{\1}{\mathbf{1}}

\newcommand{\tr}{\mathrm{tr}}

\newcommand{\Var}{\mathrm{Var}}

\newcommand{\pE}{\tilde{\E}}

\newcommand{\cB}{\mathcal B}

\newcommand{\poly}{\mathrm{poly}}

\newcommand{\mper}{\,.}
\newcommand{\mcom}{\,,}
\newcommand{\Paren}[1]{\left(#1\right)}
\newcommand{\paren}[1]{(#1)}
\newcommand{\Brac}[1]{\left[#1\right]}

\newcommand{\Norm}[1]{\left\lVert#1\right\rVert}

\newcommand{\iprod}[1]{\langle#1\rangle}
\newcommand{\Iprod}[1]{\left\langle#1\right\rangle}

\newcommand{\cA}{\mathcal A}

\newcommand{\cut}{\mathsf{cut}}
\newcommand{\FK}{\mathsf{FK}}
\newcommand{\cL}{\mathcal{L}}

\begin{document}

\newcommand{\FormatAuthor}[3]{
\begin{tabular}{c}
#1 \\ {\small\texttt{#2}} \\ {\small #3}
\end{tabular}
}
\title{Algorithms approaching the threshold for semi-random planted clique\thanks{This project has received funding from the European Research Council (ERC) under the European Union’s Horizon 2020 research and innovation program (grant agreement No 815464).}}

\author{
\begin{tabular}[h!]{ccc}
  \FormatAuthor{Rares-Darius Buhai}{rares.buhai@inf.ethz.ch}{ETH Zürich}
  \FormatAuthor{Pravesh K.\ Kothari\thanks{Supported by NSF CAREER Award \#2047933, NSF \#2211971, an Alfred P. Sloan Fellowship, and a Google Research Scholar Award.}}{praveshk@cs.cmu.edu}{CMU}
  \FormatAuthor{David Steurer}{dsteurer@inf.ethz.ch}{ETH Zürich}
\end{tabular}
} %
\date{\today}

\maketitle
\thispagestyle{empty}

\begin{abstract}
  We design new polynomial-time algorithms for recovering planted cliques in the semi-random graph model introduced by Feige and Kilian~\cite{FK01}.
The previous best algorithms for this model succeed if the planted clique has size at least \(n^{2/3}\) in a graph with \(n\) vertices~\cite{MMT20,DBLP:conf/stoc/CharikarSV17}.
Our algorithms work for planted-clique sizes approaching \(n^{1/2}\)
--- the information-theoretic threshold in the semi-random model~\cite{steinhardt2017does} and a conjectured computational threshold even in the easier fully-random model.
This result comes close to resolving open questions by Feige~\cite{bwca-semirandom} and Steinhardt~\cite{steinhardt2017does}.

To generate a graph in the semi-random planted-clique model, we first 1) plant a clique of size \(k\) in an \(n\)-vertex \Erdos--\Renyi graph with edge probability \(1/2\) and then adversarially add or delete an arbitrary number edges not touching the planted clique and delete any subset of edges going out of the planted clique.
For every $\epsilon>0$, we give an $n^{O(1/\epsilon)}$-time algorithm that recovers a clique of size \(k\) in this model whenever $k \geq n^{1/2+\epsilon}$.
In fact, our algorithm computes, with high probability, a list of about \(n/k\) cliques of size \(k\) that contains the planted clique.
Our algorithms also extend to arbitrary edge probabilities $p$ and improve on the previous best guarantee whenever $p \leq 1-n^{-0.001}$.

Our algorithms rely on a new conceptual connection that translates certificates of upper bounds on biclique numbers in \emph{unbalanced} bipartite \Erdos--\Renyi random graphs into algorithms for semi-random planted clique. 
Analogous to the (conjecturally) optimal algorithms for the fully-random model, the previous best guarantees for semi-random planted clique correspond to spectral relaxations of biclique numbers based on eigenvalues of adjacency matrices.
We construct an SDP lower bound that shows that the \(n^{2/3}\) threshold in prior works is an inherent limitation of these spectral relaxations.
We go beyond this limitation by using higher-order sum-of-squares relaxations for biclique numbers.

We also provide some evidence that the information-computation trade-off of our current algorithms may be inherent by proving an average-case lower bound for unbalanced bicliques in the low-degree polynomial model.
\end{abstract}

\clearpage
\microtypesetup{protrusion=false}
\tableofcontents{}
\microtypesetup{protrusion=true}
\thispagestyle{empty}

\clearpage

\pagestyle{plain}
\setcounter{page}{1}

\section{Introduction}%
\label{sec:introduction}

Clique is one of the most intensely studied combinatorial problems in theoretical computer science, both in terms of its worst-case and its average-case complexity.
It was among the first graph problems shown to be NP-complete \cite{karp1972reducibility}.
In fact, it turns out that for every \(\e>0\), it is NP-hard to find cliques of size \(n^{\e}\) even in graphs that contain cliques of size \(n^{1-\e}\) \cite{MR1687331-Haastad99,MR2403018-Zuckerman07,MR3728489-Khot14}.

The most well-studied average-case counterpart is the \emph{planted clique} problem~\cite{DBLP:journals/rsa/Jerrum92,DBLP:journals/dam/Kucera95} where the goal is to recover a $k$-clique added to an \Erdos--\Renyi random graph \(G(n,1/2)\).
Such a clique is uniquely identifiable if $k \gg 2 \log _2 n$.
There are polynomial time algorithms based on rounding the second eigenvector of the adjacency matrix~\cite{MR1662795-Alon98} as well as basic semidefinite programming relaxations (e.g., the \Lovasz theta function)~\cite{MR1742351-Feige00,MR1969394-Feige03} to recover the planted clique with high probability whenever \(k\ge n^{1/2}\).
Closing the exponential gap between the information-theoretic threshold value of $k$ and the threshold of the best known algorithms is a tantalizing open question that has inspired a large body of research, culminating in lower bounds against restricted classes of algorithms, such as statistical query algorithms~\cite{MR3664576-Feldman17} and sum-of-squares relaxations~\cite{DBLP:conf/focs/BarakHKKMP16}, that vastly generalize the current algorithms for this problem.
These concrete lower bounds provide some rigorous evidence that current algorithms for planted clique are optimal.

\paragraph{Fragility of algorithms.} 
Unfortunately, many algorithms for the planted clique problem are \emph{fragile}: a small number of adversarial changes to the input can cause the natural algorithms to break down completely.
This includes methods based on basic statistics such as degrees of vertices or eigenvalues of the adjacency matrix that provide the strongest possible guarantees for the problem.
Such fragility can be viewed as known algorithms \emph{overfitting} to the choice of the distributional model. 

In response, a significant research effort has gone into finding algorithms resilient against even the most benign forms of adversarial modifications.
This includes a long line of work on \emph{monotone adversary} models introduced in~\cite{feige1998heuristics} for average-case formulations of clique and coloring (i.e., community detection) \cite{MR3536616-Montanari16, MR3536617-Moitra16, liu2022minimax}.
In the context of planted clique, such models correspond to starting from the standard planted clique input and allowing an adversary to \emph{delete} any subset of edges not in the planted clique.
\david{pet peeve: i don't like the word ``intuitively''.
  people use it to mean very different things.
  sometimes they use it to indicate  high level explanations.
  sometimes it is used to indicate some kind of occam-razor type reasoning.
  according to the dictionary intuition is some belief or insight that cannot be explained.
  i think that is usually not what it is meant in paper (or at least it would be very unhelpful as part of some explanation).
} \pravesh{Good point. :) I reworded it a little bit. }
Such deletions are, in principle, only helpful since the planted clique continues to be the true maximum clique in the resulting graph. And indeed, while basic statistics and spectral methods fail in the presence of monotone adversaries, natural analyses of more resilient algorithms based on semidefinite programming~\cite{MR1742351-Feige00} succeed at the same $k=O(\sqrt{n})$ threshold while tolerating \emph{monotone adversaries}. 

\david{I think in this discussion it would be good to convey that the FK semirandom model is not just the natural next step after monotone adversaries.
  instead i think we should (and can!) argue that getting the same threshold for monotone adversaries is not that surprising and that, in contrast, the semirandom model is qualitatively different.
  here we should probably explain that it is provably not true that they are only helpful (because for sbm they can move the information theoretic threshold).
  however certain kind of natural analyses automatically show robustness against monotone adversaries.
} \pravesh{This is a good point...this discussion does need to be modified to show the difference between monotone adversaries vs semirandom models.}

\paragraph{Semi-random model.}
A seminal work by Feige and Kilian~\cite{FK01}  introduced the following semi-random planted clique model following the classical work of~\cite{BLUM1995204} on semi-random coloring. 
Such semi-random models combine a distributional input with a monotone adversary and an adversarial choice at the same time. 
After the introduction of this model, similar semi-random models have been studied for a wide range of combinatorial optimization problems, including graph partitioning and constraint satisfaction problems.
We refer the interested reader to the excellent survey~\cite{bwca-semirandom}.

\begin{definition}[Feige--Kilian semi-random planted clique model, $\FK(n,k,p)$]
  For \(n,k\in \N\) with \(k\le n\) and \(p \in [0, 1]\), we let \(\FK(n,k,p)\) be the collection of distributions over graphs with vertex set \(V=[n]\) sampleable by a process of the following form:
  \david{the labels for this list sound inconsistent.
    In particular, the term monotone adversary doesn't go well with the other terms: distributional choice and worst-case choice.
    not exactly sure what's the right thing to do.
    maybe we could talk about distributional, monotone, and worst case \emph{phases}.
  another thing that doesn't quite fit here is that the monotone part is of course also worst case.}
  \pravesh{One way to do this might be the following: we slightly change the description of the model to: 1) random or distributional generation, 2) Adversarial Deletion of Cut Edges, 2) Adversarial Addition of Non-Cut Edges}
\begin{enumerate} 
\item \emph{Random Generation Phase}: Choose a uniformly random subset \(S^*\subseteq V\) of size \(k\) and add a clique on \(S^*\) to an \Erdos--\Renyi random graph \(G(n,p)\) (which includes each possible edge independently at random with probability \(p\)),
\item \emph{Adversarial Deletion Phase}: delete an arbitrary subset of edges going out of \(S^*\) adaptively (i.e., possibly depending on the previous random choices),
\item \emph{Adversarial Addition Phase}: replace the subgraph induced on \(V\setminus S^*\) by an arbitrary one, again adaptively.
\end{enumerate}
\end{definition}

Unlike planted clique with monotone adversaries, semi-random models are far from ``helpful''.
In particular, the planted clique isn't necessarily the maximum clique in the resulting graph.
And the adversarial choices in the generation process are known to result in significantly altered information-theoretic thresholds at which efficient algorithms can succeed for related problems such as community detection in the stochastic block model~\cite{MR3536617-Moitra16}.
\pravesh{I added this blurb. Not sure if this is the right point to talk about it. Also needs citation.}

If \(p=1\), the above model recovers the worst-case version of the clique problem. On the other hand, by omitting the last two steps, we recover the original planted clique model, and by omitting only the last step we recover the planted clique with ``helpful'' monotone adversaries. Importantly, the last two steps are \emph{adaptive} and can be chosen adversarially in response to the first step.
In absence of adaptivity (i.e., when the last two steps are oblivious to the distributional choices), the model becomes significantly easier algorithmically. 

We write \(G\sim \FK(n,k,p)\) to denote a graph sampled according to one of the distributions in \(\FK(n,k,p)\).
For particular choices of parameters \(k=k(n)\) and \(p=p(n)\), our goal is to develop an algorithm that succeeds with high probability for \emph{every} distribution described by \(\FK(n,k,p)\).

\paragraph{What does it mean for the algorithm to succeed?}
Since the graph induced on $V \setminus S^*$ could be a worst-case hard instance for the clique problem, it is NP-hard to find a maximum clique in \(G\).
So the goal of the algorithm is to find a clique of size \(k\) in $G$.
For the original planted clique model (and the version with helpful monotone adversaries), we could with high probability recover the planted clique \(S^*\) in \(G\).
In contrast, in the semi-random model, this task is impossible information-theoretically because the adversary could simulate multiple disjoint copies of the distributional process in \(V\setminus S^*\).
Instead, we can ask the algorithm to compute a small list of (pairwise almost-disjoint) \(k\)-cliques in $G$ that contains the planted clique $S^*$. Such a list also allows uniquely identifying \(S^*\) if, in addition, we are given a random vertex of \(S^*\) as advice.

\pravesh{Removed this paragraph in the context of the new discussion since the failure of heuristics applies to monotone adversary model too.}

In their work introducing this model, ~\cite{FK01} gave an algorithm that uses a Gaussian rounding~\cite{MR1412228-Goemans95} of the vector solution for the \Lovasz theta SDP relaxation combined with  a combinatorial cleanup step to produce a correct list.
For any $p$ such that $1-p \geq (1+\epsilon) \ln (n)/n$, their algorithm works if $k \geq \alpha n$ for some constant $\alpha>0$.
Such a guarantee is essentially optimal if $1-p= \tilde{O}(1/n)$.
The main focus of subsequent works has been in the case when $1-p$ is larger.
In particular, the case of $p=1/2$ (and more generally, any constant $<1$) is of special interest.
In this case, one can ideally expect polynomial time algorithms that succeed for $k\sim \sqrt{n}$ as in the case of average-case planted clique.
We will focus on the case of $p=1/2$ in this introduction for the sake of clarity.

\paragraph{Prior work.}
Algorithms in prior works rely on rounding carefully designed semidefinite programming (SDP) relaxations.
In the slightly easier setting that drops the adversarial deletion step from the model, Charikar, Steinhardt and Valiant~\cite{DBLP:conf/stoc/CharikarSV17} gave an algorithm based on a semidefinite programming relaxation for \emph{list-decodable mean estimation} that succeeds whenever $k \geq O(n^{2/3} \log^{1/3} (n))$.
Their guarantee was improved to  $k \geq O(n^{2/3})$ by Mehta, Mackenzie and Trevisan~\cite{MMT20}.
The algorithm of~\cite{MMT20} is based on a variant of the Lovász theta SDP (that they call ``crude'' or C-SDP) with an objective function that incentivizes ``spread-out'' vector solutions and analyzed via the Grothendieck inequality.
They suggest (though don't prove) that their SDP should fail if $k = o(n^{2/3})$.
Further heightening the intrigue, Steinhardt~\cite{steinhardt2017does} proved that if $k=o(\sqrt{n} )$, then it is information-theoretically impossible to identify an $O(n/k)$-size list, indicating an \emph{information-theoretic} (as opposed to computational) phase transition at $k\sim \sqrt{n}$.

\paragraph{Feige's open question.} Given the apparent barrier for the basic semidefinite program at $k\sim n^{2/3}$, it is natural to ask: is the semi-random variant harder than the average-case planted clique problem or could there be algorithms that succeed for $k$ approaching the $O(\sqrt{n})$ threshold?
In his survey on semi-random models~\cite{bwca-semirandom}, Feige posed (see Section 9.3.4, Page 205) this as an outstanding open question and hoped for algorithms for semi-random planted clique matching the $k\sim \sqrt{n}$ threshold for the average-case variant. 

\subsection*{Results}

In this work, we nearly resolve Feige's question and give an algorithm for the semi-random planted clique problem that works for $k$ approaching $\sqrt{n}$.
Specifically, we give a scheme of algorithms that, for any $\epsilon >0$, run in time $n^{O(1/\epsilon)}$ and succeed in solving the semi-random planted clique problem whenever $k \geq n^{1/2+\epsilon}$:

\begin{theorem}[Main result, see Theorem~\ref{thm:main-result} for a detailed version]
  For every $\epsilon>0$, there is an algorithm that, given a graph $G$ on $n$ vertices as input, computes a list of vertex subsets in time $n^{O(1/\epsilon)}$ satisfying the following guarantee:
  If $G$ is generated according to $\FK(n,k,1/2)$ for $k \geq n^{1/2 + \epsilon}$,
  then with probability at least $0.99$ the algorithm outputs a list of at most $(1+o(1))\tfrac n k$ cliques of size \(k\) such that one of them is the clique planted in $G$.
\end{theorem}

In particular, our algorithm manages to recover planted cliques of size $k$ approaching $\sim \sqrt{n}$ --- the information-theoretic threshold~\cite{steinhardt2017does} and the conjectured computational threshold even for the easier fully-random planted clique problem.
This improves on the best known prior algorithm~\cite{MMT20} that gives a polynomial time algorithm that succeeds whenever $k \geq O(n^{2/3})$.

Our approach extends to edge probabilities \(p\) beyond the choice \(p=1/2\) and yields improved guarantees even when $p = 1-o(1)$, though in that case we do not approach the information-theoretic threshold.
Our hardness results (discussed below) show that such an outcome might be inevitable.

\paragraph{Higher degree sum-of-squares vs basic SDP.}
Our algorithm relies on rounding a high constant degree sum-of-squares relaxation (that maximizes a natural ``entropy-like'' objective function) of the natural integer program for finding $k$-cliques in graphs. As we discuss below, this is likely necessary as for the natural certification problem (discussed below) that arises in algorithms for semirandom planted clique, the basic SDP (Lemma~\ref{lem:sdp-lb}) has a lower bound that precludes recovering $k \ll n^{2/3}$-size cliques. 
This is in sharp contrast to the average-case planted clique problem where no constant degree strengthening of the basic SDP allows recovering planted cliques of size $k =o(\sqrt{n})$~\cite{DBLP:conf/focs/BarakHKKMP16} (i.e., asymptotically smaller than the threshold for recovery using the basic SDP). In fact, the best known analyses can only obtain an $n^{O(t)}$ algorithm that succeeds whenever $k \geq O(\sqrt{n/2^t})$~\cite{MR1969394-Feige03}.

Indeed, few combinatorial optimization problems are known to benefit from high constant degree sum-of-squares relaxations. Some notable exceptions include approximating constraint satisfaction problems on graphs with small threshold rank \cite{MR3424199-Arora15, MR2932723-Barak11} (where the high degree corresponds to the threshold rank) and approximating the maximum bisection in a graph \cite{MR3205223-Raghavendra12} (where the high degree helps deal with the cardinality constraint). Our work adds a new example to this list that appears to be more unstructured than earlier examples.

We note that, in contrast to combinatorial optimization, in statistical estimation higher degree sum-of-squares relaxations have recently been pivotal in algorithmic applications such as robust method of moments~\cite{DBLP:conf/stoc/KothariSS18}, linear regression~\cite{DBLP:conf/colt/KlivansKM18,DBLP:conf/stoc/BakshiP21}, list-decodable learning~\cite{DBLP:conf/stoc/KothariSS18,DBLP:conf/nips/KarmalkarKK19,raghavendra2020list,DBLP:conf/colt/RaghavendraY20,DBLP:conf/soda/BakshiK21,DBLP:conf/stoc/IvkovK22}, and settling the robust clusterability and learnability of high-dimensional Gaussian mixtures~\cite{DBLP:conf/stoc/KothariSS18,DBLP:conf/stoc/Hopkins018,BK20b,DBLP:conf/focs/BakshiDHKKK20,DBLP:journals/corr/abs-2012-02119,DBLP:conf/stoc/LiuM21}. %

\paragraph{Rounding and connection to certifying bicliques.}
Our rounding algorithm is reminiscent of the ``rounding by votes'' strategy employed in several recent works on list-decodable learning~\cite{DBLP:conf/nips/KarmalkarKK19,DBLP:conf/soda/BakshiK21,DBLP:conf/stoc/IvkovK22}.
Our analysis relies on a new connection to efficient certificates: we can recover a small list which includes the planted clique if we can certify that the planted clique has small intersection with all other $k$-cliques. 
This reduces to certifying upper bounds on bipartite cliques in unbalanced bipartite \emph{random} graphs: see Definition~\ref{def:cert-problem} for a standalone definition of the problem.

Given a bipartite graph $H=(U,V,E)$ with $|U|=k$, $|V|=n$, and each bipartite edge included in $E$ with probability $p$ independently, it is easy to prove by a standard application of Chernoff and union bounds that there is no $\ell$ by $k$ bipartite clique in $H$ with $\ell \gg O(\log n/(1-p))$.
The \emph{bipartite clique certification} problem asks to find a polynomial time verifiable certificate that $H$ contains no $\ell$ by $k$ biclique for $\ell$ as small as possible.  
This is a variant of the more standard biclique certificate problem (see, e.g., ~\cite{DBLP:conf/stoc/FeldmanGRVX13}) where both the graph and the cliques we are interested in are unbalanced. 

Our main result is based on the following primitive that certifies in time $n^{O(1/\epsilon)}$ that a random $k$ by $n$ bipartite graph does not contain $\ell$ by $k$ blicliques for $\ell = n^{\epsilon}$ and any $k \geq \tilde{O}(\sqrt{n})$. Our certificates are based on $O(1/\epsilon)$-degree sum-of-squares proofs and this is necessary -- we prove that for any $\ell = o(n/k)$, there are no degree $2$ sum-of-squares (i.e., basic SDP) certificates of absence of bicliques. In particular, unlike the case of balanced bipartite random graphs, the unbalanced setting seems to naturally benefit from large constant degree sum-of-squares certificates.

\begin{theorem}[Informal, see Theorem~\ref{thm:certification-density-half} and Lemma~\ref{lem:sdp-lb}]
\label{thm:certification-intro}
For every $\epsilon>0$, there is an $n^{O(1/\epsilon)}$ time algorithm that takes input a bipartite graph $H=(U,V,E)$ with $|U| = k$, $|V|=n$ and each bipartite edge included in $E$ with probability $p=1/2$, and with probability at least $0.99$ over the draw of $H$ outputs an $n^{O(1/\epsilon)}$-time verifiable certificate that $H$ contains no $\ell$ by $k$ biclique for $\ell \leq n^{\epsilon}$ whenever $k \geq \tilde{O}(\sqrt{n})$.

Further, 1) the certificate can be expressed as an $O(1/\epsilon)$ degree sum-of-squares refutation of the biclique axioms (see \eqref{eq:biclique-system}) and 2) there does \emph{not} exist a degree $2$ certificate (equivalently, based on the ``basic SDP'') to certify a bound of $\ell = o(n/k)$.
\end{theorem}

The results of prior works can be obtained by a simple spectral certificate captured by the basic SDP relaxation for upper bounding such bicliques (see detailed discussion in Section~\ref{subsec:spectral-certs} of the techniques).
In contrast, in this work, we depart from the spectral certificates and rely on a certain simple \emph{geometric} certificate based on upper bounds on the size of sets of \emph{pairwise negatively correlated vectors}.
If $\epsilon=O(\log \log n/\log n)$ is chosen such that $\ell = O(\log n)$ and $k = O(\sqrt{n \log n})$, the biclique refutation above translates into an algorithm for semi-random planted clique that works for $k \geq O(\sqrt{n \log n})$ in time $n^{O(\log n / \log \log n)}$, matching Steinhardt's information-theoretic lower bound~\cite{steinhardt2017does} and the threshold for the best known efficient algorithms for planted clique up to a $\sqrt{\log n}$ factor.

\paragraph{Hardness of refuting bicliques.}
We provide some evidence that improving on our biclique certification algorithms likely requires new techniques by proving a lower bound in the low-degree polynomial model.
The low-degree polynomial model (see~\cite{KWB19} for a great exposition) is a restricted model for statistical \emph{distinguishing} problems.
More precisely, the model considers problems where we are given a single sample (an instance of an algorithmic problem, a graph in our case) with the promise that it is an independent sample from one of two possible distributions: $D_{\mathrm{null}}$ --- a distribution that does not admit solutions, usually a natural random model, and $D_{\mathrm{planted}}$ --- a closely related distribution that does admit solutions.
Informally speaking, the low-degree model restricts distinguishers to thresholds of low-degree polynomials of the input.
While low-degree polynomials might appear restricted, they capture several algorithms including power iteration, approximate message passing, and local algorithms on graphs (cf.\ \cite{DMM09,GJW20}).
Moreover, it turns out that they are enough to capture the best known spectral algorithms for several canonical problems such as planted clique, community detection, and sparse/tensor principal component analysis~\cite{BHK19,HS17,DKWB19,HKP+17}.

This model arose naturally from work on constructing sum-of-squares lower bounds for the planted clique problem~\cite{BHK19}.
It was formalized in~\cite{HKP+17} and conjectured to imply sum-of-squares lower bounds for certain average-case refutation problems.
Subsequently, starting with~\cite{HS17} (see also~\cite{Hop18}), researchers have used the low-degree polynomial method as a technique to demarcate average-case algorithmic thresholds~\cite{HKP+17,GJW20,SW20,Wei20}.

In our case, $D_{\mathrm{null}}$ will be the $B(k,n,p)$ model: bipartite graphs with left vertex set of size $k$, right vertex set of size $n$, and each bipartite edge present in the graph with probability $p$.
Notice that if we had an algorithm that certifies the absence of $\ell$ by $k-\ell$ bicliques in such a graph then we can distinguish between $D_{\mathrm{null}}$ and any $D_{\mathrm{planted}}$ supported on bipartite graphs that admit $\ell$ by $k-\ell$ bicliques. 
Thus the distinguishing problem is formally easier than the task of certification (also known as refutation).
Despite the restrictedness of the low-degree model, we observe that for average-case planted clique with $k \gg \sqrt{n}$, \emph{constant} degree polynomials suffice to distinguish between $D_{\mathrm{null}} = G(n,1/2)$ and $D_{\mathrm{planted}} = G(n, 1/2)$ + $k$-clique. 

\begin{theorem}[Low-degree polynomial heuristic for biclique certification problem]
Fix $\epsilon>0$ small enough and $n$ large enough.
Let $D_{\mathrm{null}} = B(k,n,1/2)$ be the distribution on $(k,n)$-bipartite graphs where every edge is included independently with probability $1/2$.
For $k = n^{1/2+\epsilon}$, there is a distribution $D_{\mathrm{planted}}$ on $(k,n)$-bipartite graphs containing an $\ell$ by $k-\ell$ bipartite clique for $\ell = n^{0.1}$ such that the norm of the degree-$\Omega(1/\epsilon)$ truncated likelihood ratio between $D_{\mathrm{planted}}$ and $D_{\mathrm{null}}$ is $1+o(1)$.
\end{theorem}

Informally speaking, the above theorem asserts that, for $k=n^{1/2+\epsilon}$, statistical tests based on computing thresholds of $\Omega(1/\epsilon)$-degree polynomials fail to distinguish between $D_{\mathrm{null}}$ that does not admit $O(\log n)$ by $k$ bicliques and $D_{\mathrm{planted}}$ that contains an $n^{0.1}$ by $k$ biclique.
It turns out that the most natural planted model (plant a random $\ell$ by $k-\ell$ clique and sample the rest of the graph independently) can be distinguished from $D_{\mathrm{null}}$ using just degree $1$ polynomials and thus does not suffice to prove the above theorem.
Instead, we use an edge-adjusted model where the probability of sampling edges outside the biclique is reduced in order to make the degree distribution of left vertices match that of $B(k,n,p)$.

For $p=1/2$, the above theorem shows that we need polynomials of degree $O(1/\epsilon)$ in order to distinguish between $D_{\mathrm{null}}$ and bipartite graphs with $n^{0.1}$ by $n^{1/2+\epsilon}$ bicliques.
Given the contrast to the planted clique problem where the corresponding distinguishing problem can be solved by constant degree polynomials, we obtain some (weak) evidence that beating the guarantees of our current certificates may require new techniques for $p=1/2$.
For general $p$, a similar lower bound suggests that the degree of the polynomial required to distinguish between $D_{\mathrm{null}}$ and $D_{\mathrm{planted}}$ is larger than any function (independent of $n$) of $1/\epsilon$, or else that $k$ needs to scale with $1/(1-p)$ instead of the information-theoretical optimal scaling of $1/\sqrt{1-p}$ --- note that this discrepancy is $\poly(n)$ when $1-p=1/\poly(n)$.

\section{Techniques} \label{sec:overview}

In this section, we provide a high-level overview of our algorithm for the semi-random planted clique problem.
For simplicity of exposition, we will focus on the important case of $p=1/2$. 

Given a graph $G$ generated according to $\FK(n,k,1/2)$, our goal is to construct a small \emph{list} of candidate $k$-cliques in $G$ such that the true planted clique $S^*$ is contained in the list (we will call such lists \emph{correct}).
Our construction will also ensure that a constant fraction of the vertices in $S^*$ do not appear in any other clique in the list.
As a result, we can also uniquely recover $S^*$ with high probability when given, in addition, a uniformly random vertex in $S^*$.

Our algorithm and its analysis rely on the \emph{proofs-to-algorithms} method (see~\cite{TCS-086,BarakS16} for more on the usage of this method). 

\paragraph{Inefficient algorithm.} 
Let's first find an algorithm, even if inefficient, to generate a $\poly(n)$ size correct list, i.e., one that contains $S^*$.
Notice that simply outputting all $k$-cliques in $G$ can lead to an exponentially large (i.e., $\sim n^k$) size list since we have no control over the subgraph induced on $[n] \setminus S^*$ (e.g., consider a clique on $[n] \setminus S^*$).
Instead, we will enumerate all $k$-cliques in $G$ that satisfy an additional property such that 1) the property is satisfied by the planted $k$-clique on $S^*$ with high probability, and 2) every graph $G$ has at most $(1+o(1))n/k$ $k$-cliques satisfying the property.
This property is quite natural and asks for the bipartite graph with the $k$-clique on the left and the rest of the vertices on the right to not contain a large \emph{unbalanced} biclique with many vertices on the left side.
Recall that an $\ell$ by $r$ biclique in a bipartite graph $H$ is a set of vertices that consists of $\ell$ left vertices and $r$ right vertices such that $H$ contains all possible bipartite edges between the two sides.

\begin{definition}[Good $k$-cliques]
Let $G$ be a graph on $n$ vertices. 
We say that a $k$-clique $S$ in $G$ is $\ell$-\emph{good} if every biclique $(L,R)$ in the bipartite graph with left vertex set $S$, right vertex set $[n] \setminus S$, and edge set $\cut_G(S)$ satisfies $|L| \leq \ell$ whenever $|R| \geq 1$ and $|L|+|R|=k$. 
\end{definition}

The planted $k$-clique on $S^*$ is $O(\log n)$-good with high probability over the draw of $\cut(S^*)$.

\begin{proposition}[Bipartite clique number of $\cut(S^*)$] \label{prop:biclique-number-overview}
Let $k,n \in \N$ and $G \sim \FK(n,k,1/2)$.
Then, for large enough $n$ and a constant $c>0$, with probability at least $0.99$ over the draw of edges in $\cut(S^*)$, for any $L \subseteq S^*$, $R \subseteq [n] \setminus S^*$ such that $(L,R)$ is a biclique in $\cut(S^*)$ satisfying $|R| \geq 1$ and $|L| + |R| = k$, we have $|L| \leq c \log_2 n$. 
\end{proposition}
\begin{proof}
The proof is a simple application of the first moment method.
Note that it is enough to argue the proposition in the absence of the monotone adversary as deleting any subset of edges in $\cut(S^*)$ maintains the goodness of $S^*$.

The probability that $\cut(S^*)$ contains all the edges between $L \subseteq S^*$ and $R \subseteq [n] \setminus S^*$ is at most $2^{-(k-|L|)|L|}$.
Thus, the expected number of bicliques $(L,R)$ such that $|L| \geq c \log_2 n$ is at most $\sum_{\ell +r = k, \ell \geq c \log_2 n} {n-k \choose k-\ell} {k \choose \ell} 2^{-(k-\ell)\ell} \rightarrow 0$ as $n \rightarrow \infty$ if $c$ is a large enough constant.
The proposition then follows by an application of Markov's inequality.
\end{proof}

A simple greedy argument upper bounds the number of $\ell$-good $k$-cliques if $k \geq O(\sqrt{n \log n})$.

\begin{proposition}[Number of good $k$-cliques] \label{prop:number-good-cliques}
Let $G$ be a graph on $n$ vertices.
Then, for any $\ell$, if $k > 2\sqrt{n \ell/\delta}$ for some $\delta <1$, then the number of $\ell$-good $k$-cliques in $G$ is at most $(1+\delta)n/k$.
\end{proposition}
\begin{proof}
Suppose not and take any $m = (1+\delta)n/k$ such good $k$-cliques. 
Observe that any pair of $\ell$-good $k$-cliques $S,S'$ can only intersect in at most $\ell$ vertices, as otherwise $\cut(S)$ would contain a biclique with more than $\ell$ left vertices.
Thus, the $m$ good $k$-cliques must cover at least $mk - m^2 \ell = n + \delta n - (4n^2/k^2)\ell$ vertices, a number that exceeds the total number of vertices $n$ if $k > 2 \sqrt{n \ell/\delta}$. 
\end{proof}

Propositions~\ref{prop:biclique-number-overview} and~\ref{prop:number-good-cliques} immediately yield an $n^{O(k)}$ time algorithm to generate a correct list of $k$-cliques of size $(1+\delta)(n/k)$.
In fact, this algorithm can be made to run in time $n^{O(\log n)}$ by enumerating all $c \log_2 n$ size cliques $Q$ in $G$ and adding a $k$-clique to the list if the common neighborhood of $Q$ is of size $\geq k-|Q|$ and forms a clique with $Q$.

\subsection{Efficient algorithms and biclique certificates}

In the inefficient algorithm above a key idea is the claim that $\cut(S^*)$ does not have an $\ell$ by $k-\ell$ bipartite clique for $\ell > O(\log n)$.
Note that $\cut(S^*)$ is an unbalanced (left side is much smaller than the right) $k$ by $n-k \approx n$ bipartite graph and we proved that it does not have an unbalanced ($\gg O(\log n)$ vertices from the left) biclique in it. 

Key to our efficient algorithm for semi-random planted clique is an efficiently computable \emph{certificate} of non-existence of unbalanced bicliques in $H$ as above (i.e., a \emph{refutation}).

Let $B(n_1, n_2, p)$ denote the distribution on bipartite graphs with $n_1$ left and $n_2$ right vertices and every bipartite edge included with probability $p$ independently.
Let us phrase the version relevant to us formally before continuing:

\begin{definition}[Refuting unbalanced bicliques]%
\label{def:cert-problem}
An algorithm that takes as input a bipartite graph $H = (U, V, E)$ with $|U|=k$, $|V|=n-k$ refutes $\ell$ by $k-\ell$ bicliques in random $k$ by $n-k$ bipartite graphs if it has the following two properties:
\begin{enumerate}
  \item \textbf{Correctness:} If the algorithm outputs $s$, then there is no $s$ by $k-s$ biclique in $H$. 
  \item \textbf{Utility: } If $H \sim B(k,n-k,1/2)$, then the algorithm outputs $s \leq \ell$ with probability at least $0.99$ over the draw of $H$.
\end{enumerate}
\end{definition}

\begin{remark}[From certificates to algorithms: a heuristic] \label{remark:heuristic}
In Section~\ref{subsec:rounding}, we overview the translation of a (constant degree sum-of-squares) certificate that the left side of any size-$k$ biclique in $H\sim B(k,n-k,1/2)$ has at most $\ell$ vertices into an algorithm for the semi-random planted $k$-clique problem that succeeds whenever $k \geq O(\sqrt{n \ell})$.
This matches the simple bound in Proposition~\ref{prop:number-good-cliques} for the ``brute-force'' algorithm above.
We postpone the discussion of sum-of-squares proofs for now while noting that all certificates discussed in this section are in fact constant degree sum-of-squares certificates. 
\end{remark}

Observe that our simple analysis of the inefficient algorithm gives an $n^{O(\log n)}$ algorithm that refutes the existence of $\ell$ by $k-\ell$ bicliques in $B(k,n-k,1/2)$ with probability at least $0.99$ for $\ell = O(\log n)$.
Our goal is to find a polynomial time algorithm that succeeds for $\ell$ as close to $O(\log n)$ as possible.

The biclique refutation problem appears to be an interesting analog of refuting cliques in random (non-bipartite) graphs $G \sim G(n,1/2)$ (that underlies algorithms for the fully-random planted clique problem) or bicliques in $B(n,n,1/2)$ (i.e., the balanced bipartite graph). 
It can be thought of as \emph{certifying} the correctness of the candidates in the list that is purportedly a solution to the semi-random planted clique problem.
Finding solutions together with a certificate of correctness is an important goal by itself.
For example, this is a key advantage (in addition to tolerating a monotone adversary) of the method of Feige and Krauthgamer~\cite{MR1742351-Feige00} over the spectral algorithm~\cite{MR1662795-Alon98} for the planted clique problem.

\subsubsection{Basic spectral certificate}%
\label{subsec:spectral-certs}
Let us start by recalling the basic spectral certificate that underlies the algorithms for the average-case planted clique problem.
This certificate implicitly underlies the algorithms of~\cite{MMT20,DBLP:conf/stoc/CharikarSV17}.
Our framework translates it into an algorithm for semi-random planted clique whenever $k \gg O(n^{2/3})$. 

\begin{proposition}[Basic spectral certificate for clique number]%
\label{prop:basic-spectral-planted-clique}
In any graph $G$, the clique number $\omega(G) \leq 1+ \Norm{A}_2$ where $A$ is the $\{\pm 1\}$ adjacency matrix of $G$. 
\end{proposition}
\begin{proof}
If $x$ is a $\{0,1\}$-indicator of a $k$-clique in $G$, then, note that $k(k-1) = x^{\top}Ax \leq \Norm{x}_2^2 \Norm{A}_2 = k \Norm{A}_2$.
Thus, $k \leq 1 + \Norm{A}_2$ for any graph $G$.
\end{proof}
Thus, simply outputting the (polynomial time computable) largest singular value of $A$ gives a certificate of an upper bound on $\omega(G)$. 
Further, if $G \sim G(n,1/2)$, then standard spectral norm bounds on random symmetric $\{\pm 1\}$ matrices imply that the algorithm outputs with high probability a bound of $O(\sqrt{n})$.

Let's now see an analog of this method for bicliques. 
\begin{proposition}[Basic spectral certificate for bicliques, see Lemma~\ref{lem:simple-spectral-certificate} for a general version]
Let $H$ be the $\{\pm 1\}$ adjacency matrix of a $k$ by $n-k$ bipartite graph $H$.
For any $k$-clique in $H$, the number of left vertices $\ell$ satisfies $\ell (k-\ell) \leq \Norm{H}_2^2$.
\end{proposition}
\begin{proof}
Let $x,y$ be the $\{0,1\}$ indicators of the left and right sides of a biclique in $H$.
Then, $\Norm{x}_2^2 \Norm{y}_2^2 = x^{\top} Hy \leq \Norm{x}_2 \Norm{y}_2 \Norm{H}_2$.
Or, $(\sum_i x_i)(\sum_i y_i) = \Norm{x}_2^2 \Norm{y}_2^2 \leq \Norm{H}_2^2$. 
\end{proof}

For a random bipartite graph from $B(k,n-k,1/2)$, the $H$ is a $k$ by $n-k$ matrix with independent random $\{\pm 1\}$ entries.
For such matrices, standard results (see Fact~\ref{fact:char-matrix-spectral-norm}) show that $\Norm{H}_2 \leq O(\sqrt{k} + \sqrt{n}) = O(\sqrt{n})$.
Further, by a union bound, the degrees of all right vertices are at most $k/2 + O(\sqrt{k \log n})$ with high probability, so $k-\ell \geq k/4$ if $k \gg \log n$.
In that case, the above proposition shows that the spectral certificate refutes the existence of an $\ell$ by $k-\ell$ clique for $\ell \leq O(n/k)$. 

By applying the heuristic from Remark~\ref{remark:heuristic}, we obtain an algorithm for semi-random planted clique if $k \geq O(\sqrt{n \ell})$ with $\ell = O(n/k)$, that is, if $k \geq O(n^{2/3})$, matching the guarantees of~\cite{MMT20}.

It turns out that the bound of $\ell = O(n/k)$ based on the basic SDP/spectral relaxations is essentially tight. In Lemma~\ref{lem:sdp-lb}, we show that the basic SDP provably fails to certify that $\ell=o(n/k)$. This shows an inherent limitation of certificates based on the basic SDP/spectral relaxations. 

\paragraph{The Charikar-Steinhardt-Valiant approach.}
In their work on algorithms for list-decodable mean estimation~\cite{DBLP:conf/stoc/CharikarSV17}, the authors devised a method for the analog of the semi-random planted clique problem without the monotone adversary step.
When viewed from our vantage point of biclique refutation, their idea can be thought of as taking the $\pm 1$-neighborhood indicators of the \emph{right} hand side of the graph and treating them as $n-k$ samples of a $k$-dimensional distribution.
An $\ell$ by $k-\ell$ biclique translates~\footnote{
We note that the CSV approach directly applies to the semi-random planted clique model and does not actually yield a biclique certificate.
The reason is that an $\ell$ by $k-\ell$ biclique does not translate into non-zero mean for arbitrary bipartite graphs.
We ignore this distinction in order to allow an intuitive comparison of their technique in the context of our work.}
into the distribution having a non-zero mean.
Thus, one can apply (analogs of) list-decodable mean estimation algorithms~\cite{DBLP:conf/stoc/CharikarSV17,KothariSteinhardt17} to refute the existence of bicliques.
The guarantees of the algorithm depend on higher directional moments of the input distribution.
The ``base case'' corresponds to using just the second moments of the distribution --- and this roughly relates to the use of the basic spectral certificate above.
The higher moment variants can indeed yield improvements but this does not apply to our setting, because when seen from the vantage point of list-decodable mean estimation we have $n \ll k^2$ samples of a $k$-dimensional distribution --- a bound not sufficient for the $4$th moments to converge!
Indeed, this is the key bottleneck that leads to a barrier at $k=\tilde{O}(n^{2/3})$ for the CSV approach (and led to Steinhardt's open question for semi-random planted clique~\cite{steinhardt2017does}).

\subsubsection{Improved spectral certificates}
Can we improve on the basic spectral certificate?
We note that for related problems (e.g., densest $k$-subgraph, random constraint satisfaction, coloring random graphs) we usually get no asymptotic improvement by considering spectral certificates with larger (but polynomial size) matrices built from the instance.
Indeed, one can prove strong lower bounds~\cite{DBLP:conf/stoc/KothariMOW17,MR4399701-Jones22} that rule out such larger polynomial size certificates captured by constant degree sum-of-squares proofs. 

\paragraph{Neighborhood reduction.}
A natural way to improve the spectral certificate for the clique number of $G \sim G(n,1/2)$ from Proposition~\ref{prop:basic-spectral-planted-clique} is to cycle through all possible subsets of $t$ vertices, move to the common neighborhood of the $t$ vertices and then apply Proposition~\ref{prop:basic-spectral-planted-clique} to the induced graph on this common neighborhood.
This strategy yields an upper bound of $\omega(G) \leq t+1+\max_{S \subseteq [n], |S|=t} \Norm{A_S}_2$ where $A_S$ is the adjacency matrix of the induced subgraph on the common neighborhood of $S$.
One can prove that $\Norm{A_S}_2 \leq O(\sqrt{n/2^t})$ with high probability simultaneously for all $S$ of size $t$, certifying an upper bound of $O(\sqrt{n/2^t})$ on the clique number $\omega(G)$.
Since the resulting certificate has polynomial size only when $t=O(1)$, the improvement makes no asymptotic difference in the threshold $k$ at which polynomial time algorithms work.
As an aside, this simple certificate happens to be optimal for the degree $t$ Lovász-Schrijver SDP hierarchy~\cite{MR1969394-Feige03} applied to $G \sim G(n,1/2)$.
Repeating an analogous argument in our case also yields no asymptotic improvement unless $t=\omega(1)$ (though it does allow us to get arbitrary constant factor improvements). 

\paragraph{Tensoring.}
We consider next a natural class of ``tensoring'' schemes for producing improved spectral certificates.
Consider a bipartite graph with $\{\pm 1\}$ adjacency matrix $H'$ with the same right side but the left side containing all pairs of left vertices from $H$.
The $((i,j),k)$-th entry of $H'$ equals $H(i,k)H(j,k)$ -- the ``parity'' or product of the $\{\pm 1\}$ indicators of edges $(i,k)$ and $(j,k)$ in $H$.
$H'$ is a $k^2$ by $n$ matrix, and further, an $\ell$ by $k-\ell$ biclique in $H$ translates into an $\ell^2$ by $k-\ell$ biclique in $H'$.
The basic spectral certificate from Proposition~\ref{lem:simple-spectral-certificate} applied to $H'$ yields that $\ell^2 \leq O(\Norm{H'}_2^2/k)$. 

If $H'$ were a matrix of independent random $\{\pm 1\}$ entries, $\Norm{H'}_2 = O(\sqrt{k^2}) = O(k)$ yielding $\ell \leq O(\sqrt{k})$.
Despite $H'$ having correlations in its entries, this optimistic\footnote{
Every rectangular matrix of larger dimension $k^2$ and Frobenius norm $k\sqrt{n}$ has a spectral norm $\geq k$.}
bound is essentially correct (we will omit the proof here).
Plugging this back into our heuristic, we get an algorithm for semi-random planted $k$-clique if $k \geq O(\sqrt{n \sqrt{k}})$ or $k \gg n^{2/3}$, the same as before!
That is, even though the tensoring trick gives a different asymptotic estimate, it does not lead to any improvement in the threshold for $k$ in our semi-random planted clique application.

What happens if we ``tensor the left side'' $t$ times for $t >2$?
An optimistic estimate such as the above yields a bound of $\ell^t \leq O(k^{t-1})$ or $\ell \leq k^{1-1/t}$ -- a bound that appears to \emph{degrade} as we increase $t$!
We will omit the details here but a similarly worse bound results if we tensor the right side of $H$ instead. 

\paragraph{Two-sided tensoring beats the $n^{2/3}$ barrier but fails a long way off $\sqrt{n}$.}
It turns out simultaneously tensoring both sides unequally helps beat the $\ell \leq \max\{\sqrt{k}, n/k\}$ bound obtained via one-sided tensoring above.
Intuitively speaking, the ``optimal'' two-sided tensoring attempts to make the resulting adjacency matrix as ``square'' in dimensions as possible. 
Formal proofs require analyzing matrices of correlated random entries using the graph matrix method devised in the context of proving sum-of-squares lower bounds in~\cite{DBLP:conf/focs/BarakHKKMP16} and follow-ups.
We note without further details that two-sided tensoring appears to break down at $k \sim n^{0.61}$.

\subsection{Our certificate: bicliques imply sets of negatively correlated vectors} 
Our key idea to circumvent the bottlenecks in the natural spectral certificates is to abandon the idea of spectral certificates altogether.
Instead, we will show that a simple family of ``geometric'' certificates for biclique numbers allows us to show $\ell \leq n^{\epsilon}$ for any fixed $\epsilon >0$.
Specifically, we will show that if there is an $\ell$ by $k-\ell$ biclique in $H$, then one can extract $2^{\ell}-1$ \emph{pairwise negatively correlated vectors} in $n$ dimensions.  

In order to explain this connection, let us note a property of a random bipartite graph $H = (U, V, E) \sim B(k,n-k,1/2)$.
For any subset $S \subseteq U$ of $|S| \leq t$ vertices from the left vertex set of $H$, let $N_S(j) = \prod_{i \in S} H(i,j)$ where $H$ is the $\{\pm 1\}$-adjacency matrix of $H$.
Then $N_S$ is an $n$ dimensional vector of ``parities'' of $\{\pm 1\}$ indicators of all edges from $S$ to $\{k\}$.
Further, in a random $H$, every $N_S$ is nearly \emph{balanced}.
That is, by a simple Chernoff and union bound argument (see Lemma~\ref{lem:balacedness-density-half}), $\abs{\sum_{i \leq n-k } N_S(i)} \leq O(\sqrt{nt \log n})$ for every $S$ of size $t$. 

Let's call a $k$ by $n-k$ bipartite graph $t$-fold balanced if the above property holds: that is, every $N_S$ is approximately balanced for $|S| \leq t$.
We will now show that given an $\ell$ by $k-\ell$ biclique in a $t$-fold balanced graph, we can produce a set of ${\ell \choose t/2}$ pairwise negatively correlated vectors in $n$ dimensions. 

\begin{proposition}[Bicliques and negatively correlated vectors]%
  \label{prop:bicliques-vs-neg-corr}
  Suppose $H$ is a $k$ by $n-k$ bipartite graph that is $t$-fold balanced for some $t \in \N$.
  Suppose that $H$ contains an $\ell$ by $k-\ell$ biclique  $(L,R)$ for $k-\ell \geq k/4$.
  Then, if $k \geq O(\sqrt{nt \log n})$, there exist ${\ell \choose t/2}$ different $(n-2k+\ell)$-dimensional vectors $N_S^{-}$ (one for each $S\subseteq L$ of size $t/2$) such that $\iprod{N_S^{-}, N_T^{-}} <0$ whenever $S \neq T$.
\end{proposition}
\begin{proof}
  First observe that for any $S,T \subseteq L$ of size $t/2$, $\iprod{N_S, N_T} = \sum_{j \leq n-k} N_{S \Delta T}(j) = O(\sqrt{nt \log n})$ where we invoked the $t$-fold balancedness of $H$.
  Now, without loss of generality, assume that $R$ is the set of the first $k-\ell$ vertices on the right.
  Consider the vectors $N_S^{-}$ in $n-2k+\ell$ dimensions obtained by stripping the first $k-\ell$ coordinates off of $N_S$ for every $S \subseteq L$ of size $t/2$.
  Since $S, T\subseteq L$, the first $k-\ell$ coordinates contribute $+(k-\ell)$ to $\iprod{N_S, N_T}$.
  Thus, $\iprod{N_S^{-}, N_T^{-}} \leq O(\sqrt{nt \log n}) - k/4 <0$ if $k - \ell \geq k/4$ and $k \geq O(\sqrt{nt \log n})$.
\end{proof}

It is a standard fact that there can only be $d+1$ pairwise negatively correlated vectors in $d$ dimensions.
A weaker version can be proved via a simple argument involving quadratic polynomials over the vectors:
\begin{proposition}[Bound on negatively correlated vectors] \label{prop:spectral-bound-neg-corr}
Let $v_1, v_2,\ldots,v_N$ be $n$-dimensional vectors of length $\sqrt{n}$ each satisfying $\iprod{v_i, v_j} \leq -r$.
Then $N\leq 1+ n/r$. 
\end{proposition}
\begin{proof}
We know that $\Norm{\sum_{i\leq N} v_i}_2^2 \geq 0$.
On the other hand, $\Norm{\sum_{i = 1}^N v_i }_2^2 = \sum_{i = 1}^N \Norm{v_i}_2^2 + \sum_{i \neq j} \iprod{v_i,v_j} \leq N n - N(N-1) r$.
Putting the lower and upper bound together yields that $N-1 \leq n/r$ or $N \leq 1+ n/r$.
\end{proof}

Now, Proposition~\ref{prop:bicliques-vs-neg-corr} yields ${\ell \choose t/2}$ vectors with pairwise correlations at most $-ck$ for some constant $c>0$ if $k \gg \sqrt{nt \log n}$.
On the other hand, Proposition~\ref{prop:spectral-bound-neg-corr} yields that the number of such vectors can only be $1+O(n/k)$.
Putting these two bounds together yields that $\ell \lesssim (n/k)^{2/t}$. 
Choosing $t = 1/\epsilon$ gives us an $n^{O(1/\epsilon)}$ size certificate that $\ell$ is at most $n^{\epsilon}$. 

The above argument can be converted into a sum-of-squares refutation of bicliques in $H$ (see Theorem~\ref{thm:certification-density-half}).
The main observation is that the step where we strip the first $k-\ell$ coordinates off of $N_S$ can be done  ``within sum-of-squares'' while the remaining argument is a sum-of-squares proof by virtue of the above simple proposition. It turns out that we need some additional careful arguments to place the certificate in a usable form, which we will omit for the purpose of this overview (see Remark~\ref{rem:proof-plan}).

\subsection{From biclique certificates to algorithms for semi-random planted clique}
\label{subsec:rounding}
Our algorithms use the biclique certificates discussed previously to analyze a rounding algorithm for SDP relaxations of the standard $k$-clique axioms.
Specifically, consider the standard integer programming formulation of the $k$-clique problem written as the quadratic polynomial system $\cA=\cA(G)$ below.
Note that the solutions to $\cA(G)$ are $k$-cliques in the graph $G$ on vertex set $[n]$. 

\begin{equation} \label{eq:quadratic-formulation}
  \cA(G)\colon
  \left \{
    \begin{aligned}
      &\forall i \in [n]
      & w_i^2
      & =w_i \\
      &\forall i\in [n]
      &\textstyle\sum_{i=1}^n w_i
      &= k\\
      &\forall i,j \text{ s.t. } \{i,j\} \not\in G
      & w_i w_j
      & = 0
    \end{aligned}
  \right \}
\end{equation}  

Finding a solution to this quadratic program is clearly NP-hard.
So we will instead work with ``sum-of-squares'' SDP relaxations of the quadratic program, whose solutions can be interpreted as a generalization of probability distributions over solutions to the quadratic program.
Specifically, a degree $d$ pseudo-distribution $D$ is a relaxation of a probability distribution on $\{0,1\}^n$ in that the associated ``mass'' function can take negative values while still inheriting a non-trivial subset of the properties of probability distributions.
We will postpone the formal definition of pseudo-distributions to Section~\ref{sec:prelims} and for now note the following relevant bits:
1) Unlike an actual probability distribution, we only get access to low-degree moments (i.e., expectations of monomials) of $D$ and thus can only compute expectations of degree $\leq d$ polynomials,
2) pseudo-distributions can assign ``negative probabilities'' and thus may not assign non-negative expectations to pointwise non-negative degree $d$ polynomials $f$, but
3) degree $d$ pseudo-distributions do assign non-negative expectations to any $f$ that is a sum of squares of degree $\leq d/2$ polynomials, and
4) a pseudo-distribution of degree $d$ satisfying $\cA$ satisfies all ``low-degree inferrable'' properties of $k$-cliques but need not be supported on $w$ that indicate $k$-cliques at all.
Here, low-degree inferrable property means that for any degree $\leq d-2$ polynomial $f$ and any $\{i,j\} \not\in G$, $\pE_D [f w_i w_j] =0$.

A degree $d$ pseudo-distribution minimizing any convex objective in the pseudomoments $\pE[\prod_{i \in S} w_i]$ for $|S| \leq d$ and approximately satisfying $\cA$ at degree $d$ can be computed in time $n^{O(d)}$ (see Section~\ref{sec:preliminaries}).

Though a pseudo-distribution is not a probability distribution over solutions to $\cA$, it is still helpful for the reader to imagine it to be as such. 

How do our biclique certificates help us?
It turns out that while degree $d$ pseudo-distributions are far from actual probability distributions for $d \ll n$, they behave so for the purpose of polynomial inequalities that can be derived from $\cA$ using degree $d$ sum-of-squares proofs.
The conclusion of our biclique certificate from Proposition~\ref{prop:biclique-number-overview} can be written (see Theorem~\ref{thm:certification-density-half}) as a degree $O(t)$ consequence of the quadratic system $\cB$ (see \eqref{eq:biclique-system}) that identifies bicliques in bipartite graphs of total size $k$.
Consider the bipartite graph $\cut(S^*)$.
Let $w_L$ be the restriction of $w$ to coordinates in $S^*$ and $w_R$ be the restriction of $w$ to coordinates outside of $S^*$.
Then, $\cA$ implies that $(w_L, w_R)$ satisfy $\cB$ for the bipartite graph $\cut(S^*)$.
Since the pseudo-distribution $D$ satisfies $\cA$, we can conclude that 
\begin{equation}
\pE_{D}\left[ \Paren{\sum_{i \in S^*} w_i}^t \Paren{\sum_{i \not \in S^*} w_i}\right] \leq O(n^5/k^4) \label{eq:intro-conseq-cert}
\end{equation}
whenever the pseudo-distribution $D$ has degree at least $O(t)$.
Note that $S^*$ is not known to us but the above inequality forces the pseudo-distribution computed by the SDP to capture some non-trivial information about it.

\paragraph{The need for coverage constraints.}
Roughly speaking, \eqref{eq:intro-conseq-cert} can be interpreted as saying that the pseudo-distribution is ``supported'' only on those $w$ that cannot simultaneously appreciably intersect $S^*$ and $[n] \setminus S^*$.
Such a fact by itself seems unhelpful.
After all, the pseudo-distribution could completely ignore $S^*$ and focus on the ``worst-case" graph on $[n] \setminus S^*$.
Given the worst-case hardness of clique, the pseudo-distribution may not have any information about $k$-cliques in $[n] \setminus S^*$ and consequently the input graph.

In order to make \eqref{eq:intro-conseq-cert} useful, we must somehow ``force'' the pseudo-distribution to have a non-trivial mass on vertices in $S^*$.
Of course, we do not know $S^*$, so how can we do it?
It turns out that this can be accomplished by certain ``max coverage" constraints.
Specifically, instead of finding any pseudo-distribution consistent with $\cA$, we find one that minimizes $\Norm{\pE_{D} [w]}_2^2$.
This is a convex function of the pseudo-distribution and thus can be minimized efficiently using the ellipsoid method.
This objective forces the pseudo-distribution to be ``spread-out''. 
Indeed, in a different language, such an objective is used also in ~\cite{MMT20}, though arguably our treatment of such an objective as a max coverage constraint on sum-of-squares relaxations of $\cA$ appears to demystify the use of crude-SDP in ~\cite{MMT20}.
We note that such a max coverage constraint is at the heart of the rounding algorithms for several problems in list-decodable learning starting with~\cite{DBLP:conf/nips/KarmalkarKK19}.

A key consequence of the max coverage constraint is that, by an elementary convexity argument, it implies the following proposition:

\begin{proposition}[Max coverage pseudo-distributions]
\label{prop:max-coverage-overview}
For any pseudo-distribution $D$ on $w$ satisfying $\cA$ of degree at least $2$ and minimizing $\Norm{\pE_{D}[w]}_2^2$, we have $\sum_{i \in S^*} \pE_{D}[ w_i ] \geq \frac{k^2}{n}$. 
\end{proposition}

A rounding algorithm now falls naturally out of the above two discussions. 
We look at an $n^t$ by $n$ matrix indexed by subsets of size $t=O(1/\epsilon)$ on the rows and singleton vertices on the columns, whose value at index $(S, i)$ is $\frac{\pE_{D}[ w_S w_i ]}{\pE_{D}[w_S]}$.
Proposition~\ref{prop:max-coverage-overview} implies that the rows of this (huge) matrix corresponding to the unknown planted clique must have a large total sum.
On the other hand, as a consequence of the biclique certificate, we learn that for such rows the columns corresponding to $[n] \setminus S^*$ must have a low total contribution.
Together these two statements allow us to use a simple greedy algorithm that selects a uniformly random row of the above matrix and takes the largest $\sim k$ entries to recover a list containing a set of $\sim k$ vertices that has a large constant fraction intersection with $S^*$. Such a set can then be refined using a simple combinatorial ``cleanup'' step. 

\newcommand{\tzeta}{\tilde{\zeta}}
\newcommand{\bbQ}{\mathbb{Q}}
\newcommand{\cN}{\mathcal{N}}
\newcommand{\dtv}{\mathsf{TV}}
\section{Preliminaries}%
\label{sec:preliminaries}%
\label{sec:prelims}

We will use letters $G,H$ to denote graphs and also their $\{\pm 1\}$-entry adjacency matrices.
We adopt the convention that $G(i,j)=1$ if edge $\{i,j\}$ is present in the graph $G$.
For any $x \in \R^n$ and $S \subseteq [n]$, we use $x_S$ to denote the monomial $\prod_{i\in S} x_i$. For any $x \in \Set{0, 1}^n $, we use $|x|$ to denote $\sum_{i=1}^n x_i$.
We use the notation $O(n)$ and $\Omega(n)$ to mean an absolute constant multiplied by $n$ (in the former case, a ``large enough'' constant, and in the latter case, a ``small enough'' constant).

The bit complexity of a rational number $p/q$ is $\lceil \log_2 p \rceil + \lceil \log_2 q \rceil$.

\subsection{Sum-of-squares preliminaries}
We refer the reader to the monograph~\cite{TCS-086} and the lecture notes~\cite{BarakS16} for a detailed exposition 
of the sum-of-squares method and its usage in average-case algorithm design.
A \emph{degree-$\ell$ pseudo-distribution} is a finitely-supported function $D:\R^n \rightarrow \R$ such that $\sum_{x} D(x) = 1$ and $\sum_{x} D(x) f(x)^2 \geq 0$ for every polynomial $f$ of degree at most $\ell/2$.
We define the \emph{pseudo-expectation} of a function $f$ on $\R^d$ with respect to a pseudo-distribution $D$, denoted $\pE_{D(x)} f(x)$, as $\pE_{D(x)} f(x) = \sum_{x} D(x) f(x)$. 

The degree-$\ell$ pseudo-moment tensor of a pseudo-distribution $D$ is the tensor $\E_{D(x)} (1,x_1, x_2,\ldots, x_n)^{\otimes \ell}$ with entries corresponding to pseudo-expectations of monomials of degree at most $\ell$ in $x$.
The set of all degree-$\ell$ moment tensors of degree $d$ pseudo-distributions is also closed and convex.

\begin{definition}[Constrained pseudo-distributions]
  Let $D$ be a degree-$\ell$ pseudo-distribution over $\R^n$.
  Let $\cA = \{f_1\ge 0, f_2\ge 0, \ldots, f_m\ge 0\}$ be a system of $m$ polynomial inequality constraints.
  We say that \emph{$D$ satisfies the system of constraints $\cA$ at degree $r$} (satisfies it $\eta$-approximately, respectively), if for every $S\subseteq[m]$ and every sum-of-squares polynomial $h$ with $\deg h + \sum_{i\in S} \max\set{\deg f_i, r} \leq \ell$, $\pE_{D} h \cdot \prod_{i\in S}f_i  \ge 0$ ($\pE_{D} h \cdot \prod_{i\in S}f_i \geq \eta \cdot \norm{h}_2 \prod_{i \in S} \norm{f_i}_2 $ where $\norm{h}_2$ for any polynomial $h$ is the Euclidean norm of its coefficient vector, respectively).
  We say that $D$ satisfies (similarly for approximately satisfying) $\cA$ (without mentioning degree) if $D$ satisfies $\cA$ at degree $0$.
\end{definition}

\paragraph{Basic facts about pseudo-distributions.}
\begin{fact}[Hölder's inequality for pseudo-distributions] \label{fact:pseudo-expectation-holder}
Let $f,g$ be polynomials of degree at most $d$ in indeterminate $x \in \R^d$. 
Fix $t \in \N$.
Then, for any degree $dt$ pseudo-distribution $\tzeta$,
$\pE_{\tzeta}[f^{t-1}g] \leq \paren{\pE_{\tzeta}[f^t]}^{\frac{t-1}{t}} \paren{\pE_{\tzeta}[g^t]}^{1/t}$.
\end{fact}
Observe that the special case of $t =2$ corresponds to the Cauchy-Schwarz inequality.
The following idea of \emph{reweighted} pseudo-distributions follows immediately from definitions and was first formalized  and used in~\cite{DBLP:conf/stoc/BarakKS17}).  
\begin{fact}[Reweightings~\cite{DBLP:conf/stoc/BarakKS17}] \label{fact:reweightings}
Let $D$ be a pseudo-distribution of degree $k$ satisfying a set of polynomial constraints $\cA$ in variable $x$. 
Let $p$ be a sum-of-squares polynomial of degree $t$ such that $\pE[p(x)] \neq 0$.
Let $D'$ be the pseudo-distribution defined so that for any polynomial $f$, $\pE_{D'}[f(x)] = \pE_{D}[p(x)f(x)]/\pE_{D}[p(x)]$.
Then, $D'$ is a pseudo-distribution of degree $k-t$ satisfying $\cA$. 
\end{fact}

\paragraph{Sum-of-squares proofs.}
A \emph{sum-of-squares proof} that the constraints $\{f_1 \geq 0, \ldots, f_m \geq 0\}$ imply the constraint $\{g \geq 0\}$ consists of sum-of-squares polynomials $(p_S)_{S \subseteq [m]}$ such that $g = \sum_{S \subseteq [m]} p_S \cdot \Pi_{i \in S} f_i$.

We say that this proof has \emph{degree $\ell$} if for every set $S \subseteq [m]$, the polynomial $p_S \Pi_{i \in S} f_i$ has degree at most $\ell$ and write: 
\begin{equation*}
  \{f_i \geq 0 \mid i \leq r\} \sststile{\ell}{}\{g \geq 0\}
  \mper
\end{equation*}

\begin{fact}[Soundness]
  \label{fact:sos-soundness}
  If $D$ satisfies $\cA$ for a degree-$\ell$ pseudo-distribution $D$ and there exists a sum-of-squares proof $\cA \sststile{r'}{} \cB$, then $D$ satisfies $\cB$ at degree $rr' +r'$.
\end{fact}

\begin{definition}[Total bit complexity of sum-of-squares proofs]
Let $f_1, f_2, \ldots, f_m$ be polynomials in indeterminate $x$ with rational coefficients. 
For a polynomial $g$ with rational coefficients, we say that $\{f_1 \geq 0, \ldots, f_m \geq 0\}$ derives $\{g\geq 0\}$ in degree $k$ and total bit complexity $B$ if $g = \sum_{S \subseteq [m]} p_S \cdot \Pi_{i \in S} f_i$ where each $p_S$ is a sum-of-squares polynomial of degree at most $k-\sum_{i \in S}\operatorname{deg}(f_i)$ for every $S$, and the total number number of bits required to describe all the coefficients of all the polynomials $f_i, g, p_S$ is at most $B$.
\end{definition}

There's an efficient separation oracle for moment tensors of pseudo-distributions that allows approximate optimization  of linear functions of pseudo-moment tensors approximately satisfying constraints.
The \emph{degree-$\ell$ sum-of-squares algorithm} optimizes over the space of all degree-$\ell$ pseudo-distributions that approximately satisfy a given set of polynomial constraints:

\begin{fact}[Efficient optimization over pseudo-distributions \cite{MR939596-Shor87,parrilo2000structured,MR1748764-Nesterov00,MR1846160-Lasserre01}]
\label{fact:eff-pseudo-distribution}
Let $\eta>0$.
There exist an algorithm that for $n, m\in \N$ runs in time $(n+ m)^{O(\ell)} \poly \log 1/\eta$, takes input an explicitly bounded and satisfiable system of $m$ polynomial constraints $\cA$ in $n$ variables with rational coefficients and outputs a level-$\ell$ pseudo-distribution that satisfies $\cA$ $\eta$-approximately.
\end{fact}

\paragraph{Basic sum-of-squares proofs.}
\begin{fact}[Operator norm bound]
\label{fact:operator_norm}
Let $A$ be a symmetric $d\times d$ matrix with rational entries with numerators and denominators upper-bounded by $2^B$ and $v$ be a vector in $\mathbb{R}^d$. 
Then, for every $\epsilon \geq 0$, 
\[
\sststile{2}{v} \Set{ v^{\top} A v \leq \|A\|_2\|v\|^2_2 + \epsilon}
\]
Further, the total bit complexity of the sum-of-squares proof is $\poly(B,d,\log 1/\epsilon)$.
\end{fact}

\begin{fact}[SoS Hölder's inequality] \label{fact:sos-holder}
Let $f_i,g_i$ for $1 \leq i \leq s$ be indeterminates. 
Let $p$ be an even positive integer.
Then, 
\[
\sststile{p^2}{f,g} \Set{  \Paren{\frac{1}{s} \sum_{i = 1}^s f_i g_i^{p-1}}^{p} \leq \Paren{\frac{1}{s} \sum_{i = 1}^s f_i^p} \Paren{\frac{1}{s} \sum_{i = 1}^s g_i^p}^{p-1}}\mper
\]
Further, the total bit complexity of the sum-of-squares proof is $s^{O(p)}$. 
\end{fact}
Observe that using $p = 2$ yields the SoS Cauchy-Schwarz inequality. 

\begin{fact}[SoS almost triangle inequality] \label{fact:sos-almost-triangle}
Let $f_1, f_2, \ldots, f_r$ be indeterminates. Then,
\[
\sststile{2t}{f_1, f_2,\ldots,f_r} \Set{ \Paren{\sum_{i\leq r} f_i}^{2t} \leq r^{2t-1} \Paren{\sum_{i =1}^r f_i^{2t}}}\mper
\]
Further, the total bit complexity of the sum-of-squares proof is $r^{O(t)}$.
\end{fact}

\begin{fact}[SoS AM-GM inequality, see Appendix A of~\cite{MR3388192-Barak15}] \label{fact:sos-am-gm}
Let $f_1, f_2,\ldots, f_m$ be indeterminates. Then, 
\[
\Set{f_i \geq 0\mid i \leq m} \sststile{m}{f_1, f_2,\ldots, f_m} \Set{ \Paren{\frac{1}{m} \sum_{i =1}^m f_i }^m \geq \Pi_{i \leq m} f_i} \mper
\]
Further, the total bit complexity of the sum-of-squares proof is $\exp(O(m))$.
\end{fact}

\begin{fact}[Cancellation within sum-of-squares, Lemma 9.3 in~\cite{bakshi2020outlierrobust}] \label{fact:cancellation-within-SoS}
Let $a,C$ be indeterminates. Then,
\[
\{a \geq 0\} \cup \{a^t \leq Ca^{t-1}\} \sststile{2t}{a,C} \Set{a^{2t} \leq C^{2t}}\mper 
\]
Further, the total bit complexity of the sum-of-squares proof is $\exp(O(t))$.
\end{fact}

\begin{fact}[Univariate sum-of-squares proofs] \label{fact:univariate}
Let $p$ be a degree-$d$ univariate polynomial with rational coefficients of bit complexity $B$ such that $p(x) \geq 0$ for every $x \in \R$.
Then, for every $\epsilon>0$, there is a degree-$d$ sum-of-squares polynomial $q(x)$ with coefficients of bit complexity $O(\poly(B, \log 1/\epsilon))$ such that $\epsilon+p(x)=q(x)$.
\end{fact}

\begin{lemma}[Simple cancellation within sum-of-squares] \label{lem:simple-cancel}
Let $a$ be an indeterminate and $C$ be some positive constant. Then,
\begin{enumerate}
  \item 
  \[
\{a^2 \leq Ca \} \sststile{2}{a} \Set{a^{2} \leq C^{2}}\mper 
\]
\item 
  \[
\{a^2 \leq C \} \sststile{2}{a} \Set{a \leq \sqrt{C}}\mper 
\]
\end{enumerate}
The total bit complexity of the sum-of-squares proofs is $\poly(C)$.
\end{lemma}
\begin{proof}
For the first claim, we have:
\[
\{a^2 \leq Ca \} \sststile{2}{a} \Set{a^2 \leq a^2 + (a-C)^2 = C^2+2a^2 - 2aC \leq C^2}\mper 
\]

For the second claim, note that it is enough to prove the claim for $C=1$ (and apply this special case to $a/C$).
Using the fact that $\sststile{2}{a} \Set{(1+a)^2 \leq 2 a^2 + 2}$, we have:
\[
\{a^2 \leq 1 \} \sststile{2}{a} \Set{a = \frac{1}{4} (a+1)^2-\frac{1}{4}(1-a)^2 \leq \frac{1}{2}(a^2 +1) \leq 1}\mper 
\]
\end{proof}

We also need the following fact about random matrices:

\begin{fact}[Singular values of random matrices, consequence of Theorem 2.3.21~\cite{Tao12}]
\label{fact:char-matrix-spectral-norm}
Fix any $\epsilon>0$.
Let $A$ be a $k \times n$ matrix for $k \leq n$ with independent entries with magnitude at most $n^{0.5-\epsilon}$, mean $0$ and variance $1$. 
Then, for large enough $n$, with probability at least $0.99$, $\Norm{A}_2 \leq O(\sqrt{n})$.
\end{fact}

\begin{fact}[Singular values of rectangular random matrices, consequence of Theorem 4.5.1~\cite{vershynin2018high}] \label{fact:vershynin}
Let $A$ be a $k \times n$ matrix for $k \leq n$ with independent entries chosen uniformly from $\{-1,1\}$. Then, with probability at least $0.99$ over the draw of entries of $A$, the largest singular value of $A$ is at most $O(\sqrt{n})$ and the $k$-th smallest singular value of $A$ is at least $\Omega(\sqrt{n} - \sqrt{k-1})$. 
\end{fact}

\section{Certifying biclique bounds in unbalanced random bipartite graphs}
\label{sec:cert-half}

In this section, we develop low-degree sum-of-squares certificates of upper bounds on biclique sizes in unbalanced random bipartite graphs.
We use $H = (U, V, E)$ to denote a bipartite graph with left vertex set $U$, right vertex set $V$, and edge set $E$.

For a bipartite graph $H=(U, V, E)$, let $\cB =\cB(H)$ be the following system of polynomial constraints, which has as solution every biclique $(S,T)$ in $H$ of total size $k$ with $S=\{u \in U \mid x_u=1\}$ and $T = \{v \in V \mid y_v = 1\}$:

\begin{equation} \label{eq:biclique-system}
  \cB(H)\colon
  \left \{
    \begin{aligned}
      &\forall u \in U
      & x_u^2
      & =x_u \\
      &\forall v \in V
      & y_v^2
      & =y_v \\
      &
      & |x| + |y| 
      &= k\\
      &\forall u \in U,v \in V \text{ s.t. } \{u,v\} \not\in E
      & x_u y_v
      & = 0
    \end{aligned}
  \right \}\mper
\end{equation}  
\begin{remark}
Notice that the biclique formulation above places a constraint on the total size of the clique.
This is the natural formulation that arises in our reduction from the semi-random planted clique problem.
Intuitively, given a graph $G \sim \FK(n,k,1/2)$, the bipartite graph we care about is $H=\cut(S^*)$ where $S^*$ is the planted clique of size $k$.
The  bicliques we want to refute are obtained by taking an arbitrary $k$-clique $S$ in $G$ and looking at the induced biclique in $H$ with left vertices $S \cap S^*$ and right vertices $S \setminus S^*$.
In particular, notice that the total size of the biclique is $k$ and as long as $S \neq S^*$ the right hand side of the clique contains at least one vertex.
For more a detailed commentary, we direct the reader to Section~\ref{sec:overview}.
\end{remark}

\david{it is a bit surprising that we care about the total size of the biclique. is there a discussion about that issue somewhere earlier. maybe refer to that discussion here. or add a discussion about it somewhere.}

\pravesh{I think there's a place in the overview that we discuss it. Can also add a remark here. It's natural to care about the total size of the clique in our application (since the "left" hand side of the clique we care about is the intersection of a purported $k$ clique with the planted clique while the RHS is the ``rest''.  The ``hard''/interesting part, so to speak, is in the case when RHS of the clique is $> k/2$. This is because it's easy to rule out bicliques that use $>k/2$ vertices on the left by observing that the max degree of any vertex on the RHS is at most $k/2 + O(\sqrt{k \log n})$. This is also related to the issue that the same power of $X$ appears on LHS and RHS in the vanilla certificate but then using the above observation allows one to fix it.}

For ease of exposition, we will present our certificates and analysis for the most important case of $p=1/2$ first and then follow it up with a generalization to arbitrary $p$ in the following subsection. 

\subsection{The case of $p=1/2$} \label{subsec:certificate-half}

\david{i think the following sentence is redundant in the context of the text above.
  i think we could try to explain to the reader at this point how to think about the r parameter in the theorem below (like we are currently doing in the remark after the theorem).
  i think it is better to do these things before stating the theorem.}

The goal of the following theorem is to show that with high probability over the draw $H \sim B(k,n-k,1/2)$ of a bipartite \Erdos-\Renyi random graph with edge density $p=1/2$ and $k \leq n$, there is a degree $r$ (and thus verifiable in time $n^{O(r)}$) sum-of-squares certificate of (informally speaking) the fact that any $\ell$ by $k-\ell$ biclique with $k-\ell \geq 1$ satisfies $\ell \leq \poly(r) \cdot (n/k)^{O(1/r)}$.
In particular, for any $\epsilon>0$, by choosing $r=O(1/\epsilon)$, we get an $n^{O(1/\epsilon)}$-time verifiable certificate of the absence of $n^{\epsilon} \times (k-n^{\epsilon})$-bicliques in $H$.
Formally, we will prove:

\begin{theorem}[Sum-of-squares certificates for unbalanced bicliques in random bipartite graphs]
  \label{thm:certification-density-half}
  \david{instead of describing the random model again, i think it would be better to say that \(H\sim H(k,n,1/2)\) is a \(k\)-by-\(n\) bipartite Erdos-Renyi graph with edge probability 1/2.
  I think for this theorem it is also more natural to use n instead of n-k as the size of the right vertex set.}
Let $H \sim B(k,n-k,1/2)$ with $k \leq n$ be a bipartite \Erdos-\Renyi graph with edge probability $1/2$.
\david{i think it would be better to upper bound r in term of k and n as opposed to lower bounding k in terms of n and r.}
Then, for any $r \leq O(\frac{k^2}{n \log n})$, with probability at least $0.99$ over the draw of $H$, the sizes of the sets indicated by \(x\) and \(y\) respectively satisfy
\david{not sure about the commas in the previous sentence}
\begin{equation*}
  \cB(H) \sststile{4r+2}{x,y} \Biggl\{ |x|^{4r} |y|
  \leq (1000r)^{10r}n \Paren{\frac{n}{k}}^4 \Biggr\}\mper
\end{equation*}
Further, the total bit complexity of the sum-of-squares proof is $n^{O(r)}$.
\end{theorem}
 As a corollary, we obtain that with high probability over the choice of $H \sim B(k,n-k,1/2)$, for every pseudo-distribution $D$ of degree at least $4r+2$ satisfying $\cB(H)$, we must have $\pE_{D}[|x|^{4r} |y|] \leq (1000r)^{10r} n (n/k)^4$.

\begin{remark}
  \david{i think it is not necessary to talk about pseudo distributions here. it is clear that an sos proof implies that the inequality holds for every biclique.}
Observe that if $H$ contains an $\ell$ by $k-\ell$ biclique for $k-\ell \geq 1$ then the above theorem yields that $\ell^{4r} (k-\ell)\leq |x|^{4r}|y|\leq (1000r)^{10r} n(n/k)^4$ and thus, $\ell \leq \poly(r) \cdot n^{O(1/r)}$. That is, there exist degree $O(r)$ certificates of absence of $\ell$ by $k-\ell$ bicliques in $H$ for $\ell \sim n^{O(1/r)}$. 

\end{remark}

Our proof of Theorem~\ref{thm:certification-density-half} uses two simple pseudorandom properties of the graph $H$ and thus works for all graphs that satisfy these properties.
For every $S \subseteq U$, let $u_S$ be a vector in $\{-1,1\}^{|V|}$ so that 
$u_S(j) = \prod_{i \in S} H(i, j)$.
Then, we will need the following $r$-fold balancedness property that informally asks that the vectors $u_S$ be nearly balanced for all subsets $S \subseteq U$ of size at most $r$.
Additionally, we will need that every vertex on the \emph{right} side of $H$ has degree no larger than $k/2+O(\sqrt{k \log n})$. 

\begin{definition}[Balancedness]
  Let $H = (U, V, E)$ be a bipartite graph.
  For every $S \subseteq U$, let $u_S$ be the $|V|$-dimensional vector defined by setting $u_S(j)= \prod_{i \in S} H(i,j)$.
  Then, we say that $H$ has \emph{$r$-fold balancedness} $\Delta_r$ if, for all $S \subseteq U$ of size $|S| \leq r$, it holds that $|\sum_{j \in V} u_S(j)|\leq \Delta_r$.
\end{definition}

The following lemma verifies that the two pseudorandom properties hold for random bipartite graphs by a simple application of Hoeffding's inequality and union bounds.

\begin{lemma}[Balancedness of random bipartite graphs]\label{lem:balacedness-density-half}
Let $H = (U,V,E) \sim B(k, n-k, 1/2)$.
Then, for any $r \leq |U|$, with probability at least $0.99$ over the draw of $H$,
1) $H$ has $r$-fold balancedness $O(\sqrt{r n \log k})$, and
2) the maximum degree of a vertex in $V$ is at most $k/2 + O(\sqrt{k\log n})$.
\end{lemma}
\begin{proof}
For $S \subseteq U$ such that $|S| \leq r$, we have that $u_{S}(j)$ has mean $0$ and is bounded between $-1$ and $1$.
Then, by Hoeffding's inequality, 
\[\Pr\Brac{ \left| \sum_{j \in V} u_{p,S}(j)u_{p,T}(j) \right| \geq t \sqrt{|V|}} \leq 2e^{-t^2/2} \mcom \]
so, by a union bound over all choices of $S$,
\[\Pr\Brac{ \exists S \subseteq U \text{ s.t. } |S| \leq r, \left| \sum_{j \in V} u_{p,S}(j)u_{p,T}(j) \right| \geq t \sqrt{|V|}} \leq |U|^{r} \cdot 2e^{-t^2/2} \mper \]
Choosing $t = O(\sqrt{r \log U})$ makes the right-hand side a small constant, so we have $r$-fold balancedness $O(\sqrt{r|V| \log |U|})$.

The degree of a vertex in $V$ is a binomial random variable $\operatorname{Bin}(k, 1/2)$, which by standard bounds is larger than $k/2+t$ with probability at most $e^{-t^2/k}$.
By a union bound over all vertices in $V$, the maximum degree is larger than $kp+t$ with probability at most $|V|e^{-t^2/k}$, so choosing $t=O(\sqrt{k \log |V|})$ makes the probability a small constant.
Hence, the maximum degree is at most $k/2+O(\sqrt{k\log |V|})$.
\end{proof}

The key component of the proof of Theorem ~\ref{thm:certification-density-half} is the following lemma that gives a sum-of-squares certificate of an upper bound on a quantity closed related to $|x|^{4r}|y|$.

\begin{lemma}
\label{lem:biclique-certificate}
Let $H = (U, V, E)$ be a bipartite graph with $|U|=k$ and $|V|=n-k$ and $2r$-fold balancedness $\Delta_{2r}$.
Then, 

\begin{equation}
\label{eq:biclique-main-2}
\cB(H) \sststile{4r}{x,y} \Biggl\{  \Paren{\sum_{|S|=r} x_S}^{2} |y| \leq 
  n \Paren{\sum_{|S|=r} x_S} +  \Delta_{2r} \Paren{\sum_{|S|=r} x_S}^{2} 
\Biggr\}\mper 
\end{equation}
Further, the total bit complexity of the sum-of-squares proof is $n^{O(r)}$.

\end{lemma}

\begin{remark}[Proof plan]
\label{rem:proof-plan}
In order to interpret this lemma, we suggest the readers to think of $\sum_{|S|=r} x_S \approx |x|^{r}$ (this is formally shown to be fine in Lemma~\ref{lem:power-vs-sets}).
Then, rearranging the conclusion of the lemma yields a statement of the form $|x|^{2r} (|y|-\Delta_{2r}) \leq |x|^r n$. 
At this point, ``in real life'' (as opposed to within the sum-of-squares proof system), we could reason as follows: if $|y| > \Delta_{2r}$, then ``canceling'' $|x|^{r}$ from both sides and ``dividing through'' by $(|y|-\Delta_{2r})$ yields that $|x|^r \leq n$, giving us a bound on the left hand side of biclique as desired. On the other hand, if $|y| \leq \Delta_{2r}$, then for $k \gg \Delta_{2r}$ we have $|x| \gg k/2$, which can be ruled out by the upper bound on the maximum degree of a vertex on the right side.

This argument, however, is not easy to implement within the low-degree sum-of-squares proof system because of the case analysis involved. 
Indeed, a similar issue arises in the context of list-decodable learning and robust clustering algorithms that rely on certifiable anticoncentration (see overview of~\cite{BK20b} for a discussion and a general resolution, and also the discussion on the need for \emph{a priori} bounds in~\cite{DHKK20}).
In our situation, we can resolve this need using a more straightforward observation (see Lemma~\ref{lem:lower-bound-LHS-trick}).
The rest of the steps above can indeed by implemented within low-degree sum-of-squares via cancellation inequalities (e.g., see Fact~\ref{fact:cancellation-within-SoS}).
\end{remark}

We postpone the proof of this key lemma and first show how to use it.
The following simple lemma uses a bound on the degree of the right vertices in $H$ in order to lower bound the LHS of the conclusion of Lemma~\ref{lem:biclique-certificate}.
This will allow us to eliminate the term $\Delta_{2r} \Paren{\sum_{|S|=r} x_S}^2$ from the RHS of the conclusion of Lemma~\ref{lem:biclique-certificate}.

\begin{lemma}[Lower bounding the LHS of \eqref{eq:biclique-main-2}]
\label{lem:lower-bound-LHS-trick}
Let $H = (U, V, E)$ be a bipartite graph with $|U|=k$ and $|V|=n-k$ and maximum degree of a vertex in $V$ at most $k/2+\Delta_{\ell}$.
Then,  for any $j \in V$, we have:

\begin{equation*}
\cB(H) \sststile{4r+2}{x,y} \Biggl\{ y_j \Paren{\sum_{S: |S|=r} x_S}^2 |y| \geq \Paren{\frac{k}{2}-\Delta_{\ell}} y_j \Paren{\sum_{S: |S|=r} x_S}^2  \Biggr\}\mper
\end{equation*}
Further, the total bit complexity of the sum-of-squares proof is $n^{O(r)}$.

\end{lemma}

\begin{proof}
Every $j \in V$ has degree at most $\frac{k}{2}+\Delta_{\ell}$.
Thus, we have using the constraint system $\cB(H)$
\[
\cB(H) \sststile{2}{x,y} \Set{ y_j |x| = \sum_{i \in U : \set{i, j} \in E} x_i y_j \leq \Paren{\frac{k}{2}+\Delta_{\ell}} y_j}\mper
\]

Thus, $\cB(H) \sststile{2}{x,y} \Set{|x| = k-|y|}$ allows us to conclude:
\[
\cB(H) \sststile{2}{x,y} \Set{ y_j |y| \geq \Paren{\frac{k}{2}-\Delta_{\ell}} y_j}
\]
Multiplying both sides by the sum-of-squares polynomial $\Paren{\sum_{S: |S|=r} x_S}^2$ completes the proof.
\end{proof}

As a direct consequence of Lemma~\ref{lem:biclique-certificate} and Lemma~\ref{lem:lower-bound-LHS-trick}, we obtain:
\begin{lemma}\label{lem:biclique-certificate-final}
Let $H = (U, V, E)$ be a bipartite graph with $|U|=k$ and $|V|=n-k$ and $2r$-fold balancedness $\Delta_{2r}$ and maximum degree of a vertex in $V$ at most $k/2+\Delta_{\ell}$. 
Suppose further that $\frac{k}{2} - \Delta_{\ell} - \Delta_{2r} \geq \frac{k}{4}$.
Then, we have:

\begin{equation*}
\cB(H) \sststile{4r+2}{x,y} \Biggl\{  \Paren{\sum_{S: |S|=r}x_S}^4 |y| \leq n\Paren{\frac{4n}{k}}^4
\Biggr\}\mper 
\end{equation*}
Further, the total bit complexity of the sum-of-squares proof is $n^{O(r)}$.
\end{lemma}

\begin{proof}
We first multiply both sides of the conclusion of Lemma~\ref{lem:biclique-certificate} with the sum-of-squares polynomial $y_j^2$ for an arbitrary $j \in V$:

\begin{equation*}
\cB(H) \sststile{4r+2}{x,y} \Biggl\{  y_j\Paren{\sum_{|S|=r} x_S}^{2} |y| \leq 
  n y_j \Paren{\sum_{|S|=r} x_S} +  \Delta_{2r} y_j \Paren{\sum_{|S|=r} x_S}^{2}
\Biggr\}\mper 
\end{equation*}

Next, we use Lemma~\ref{lem:lower-bound-LHS-trick} to replace the left-hand side by a useful lower bound: 

\begin{equation*}
\cB(H) \sststile{4r+2}{x,y} \Biggl\{   \Paren{\frac{k}{2}-\Delta_\ell} y_j \Paren{\sum_{S: |S|=r} x_S}^2   \leq 
  n y_j \Paren{\sum_{|S|=r} x_S} +  \Delta_{2r} y_j  \Paren{\sum_{|S|=r} x_S}^{2}
\Biggr\}\mper 
\end{equation*}

We then move the second term on the right-hand side to the left-hand side and use that $\frac{k}{2} -\Delta_{\ell}- \Delta_{2r} \geq \frac{k}{4}$ to conclude:

\begin{equation*}
\cB(H) \sststile{4r+2}{x,y} \Biggl\{  y_j \Paren{\sum_{S: |S|=r} x_S}^2  \leq 
\frac{4n}{k} y_j \Paren{\sum_{|S|=r} x_S} 
\Biggr\}\mper 
\end{equation*}

We finally apply Lemma~\ref{fact:cancellation-within-SoS} with $a = y_j (\sum_{S: |S|=r} x_S)$, $C = \frac{4n}{k}$, and $t=2$ to obtain:
\begin{equation*}
\cB(H) \sststile{4r+2}{x,y} \Biggl\{  y_j \Paren{\sum_{S: |S|=r}x_S}^4 \leq y_j^4 \Paren{\sum_{S: |S|=r}x_S}^4 \leq \Paren{\frac{4n}{k}}^4
\Biggr\}\mper 
\end{equation*}

Summing up as $j$ varies over $V$ completes the proof.

\end{proof}

Finally, we invoke the following simple observation that allows us to replace $\sum_{|S| =r} x_S$ by $|x|^r$:
\begin{lemma}[]  \label{lem:power-vs-sets}
For every $\epsilon>0$, there is a sum-of-squares proof with coefficients of bit complexity $O(\poly(|U|, \log 1/\epsilon))$
\begin{equation} \label{eq:power-vs-sets}
\Biggl\{ x_i^2 = x_i \text{ } \forall i \in U \Biggr\} \sststile{r}{x} \Biggl\{\frac{1}{2^r r!} \Paren{\sum_i x_i}^r - \frac{2r^r}{r!} - \epsilon \leq \sum_{|S|=r} x_S \leq \frac{1}{r!} \Paren{\sum_i x_i}^r + \epsilon
\Biggr\}\mper
\end{equation}
\end{lemma}
\begin{proof}
The following polynomial identity holds: $\sum_{|S|=r} x_S = \frac{1}{r!}\Paren{\sum_i x_i}\Paren{\sum_i x_i - 1}\cdots\Paren{\sum_i x_i - (r - 1)}$. Then $\sum_{|S| = r} x_S \leq \frac{1}{r!}\Paren{\sum_i x_i}^r$ and 
\begin{align*}
\sum_{|S| = r} x_S \geq \frac{1}{r!}\Paren{\sum_i x_i - r}^r - \frac{r^r}{r!} \geq \frac{1}{2^r r!}\Paren{\sum_i x_i}^r - \frac{2r^r}{r!}\mcom
\end{align*}
where in the last inequality the subtracted term makes the inequality trivial unless $\sum_i x_i \geq 2r$, case in which we use that $\sum_i x_i -r \geq \sum_i x_i / 2$.

Notice that $\sum_{|S|=r} x_S$ is a univariate degree-$r$ polynomial in $\sum_i x_i$.
Then the inequalities $\sum_{|S| = r} x_S \leq \frac{1}{r!}\Paren{\sum_i x_i}^r$ and $\sum_{|S| = r} x_S \geq \frac{1}{2^r r!}\Paren{\sum_i x_i}^r - \frac{2r^r}{r!}$ can be written as univariate polynomial inequalities $p_U(\sum_i x_i) \geq 0$ and $p_L(\sum_i x_i) \geq 0$, respectively, with $p_U$ and $p_L$ of degree at most $r$.
It is easy to check that the coefficients of $p_U$ and $p_L$ have bit complexity $O(\poly(|U|))$, so by Fact~\ref{fact:univariate} the conclusion follows.
\end{proof}

We can finish the proof of Theorem~\ref{thm:certification-density-half} from here:

\begin{proof}[Proof of Theorem~\ref{thm:certification-density-half}]

From Lemma~\ref{lem:power-vs-sets}, we have:
\begin{equation*}
\Biggl\{ x_i^2 = x_i \text{ } \forall i \in U \Biggr\} \sststile{r}{x} \Biggl\{\frac{1}{2^r r!} |x|^r - \frac{2r^r}{r!} - \epsilon \leq \sum_{|S|=r} x_S \Biggr\}\mper
\end{equation*}

Setting $\epsilon =1$ and using that $\{0 \leq a \leq C\} \sststile{4}{a,C} \{a^4 \leq C^4\}$, we have:
\begin{equation*}
\Biggl\{ x_i^2 = x_i \text{ } \forall i \in U \Biggr\} \sststile{4r}{x} \Biggl\{|x|^{4r} \leq (100r)^{10r} + (100r)^{10r} \Paren{\sum_{|S|=r} x_S}^{4} \Biggr\}\mper
\end{equation*}

Now we want to combine this with the conclusion of Lemma~\ref{lem:biclique-certificate-final}.
We briefly verify that we satisfy the condition $\frac{k}{2} - \Delta_{\ell} - \Delta_{2r} \geq \frac{k}{4}$.
We have by Lemma~\ref{lem:balacedness-density-half} that $\Delta_{\ell} = O(\sqrt{k \log |V|}) = O(\sqrt{k \log n})$ and $\Delta_{2r} = O(\sqrt{r|V| \log |U|}) = O(\sqrt{rn \log n})$.
Observe that for $k \geq O(\sqrt{rn \log n})$ large enough the condition is satisfied.
Then we have:
\begin{equation*}
\cB(H) \sststile{4r+2}{x,y} \Biggl\{ |x|^{4r} |y| \leq 
 (100r)^{10r} |y| + (100r)^{10r} n\Paren{\frac{4n}{k}}^4 \Biggr\}\mper 
\end{equation*}

Observing that $\cB(H) \sststile{2}{y} \Set{ |y| \leq n}$ completes the proof. 

\end{proof}

\paragraph{Proof of Lemma~\ref{lem:biclique-certificate}.}
We now return to the proof of Lemma~\ref{lem:biclique-certificate}.
\begin{proof}[Proof of Lemma~\ref{lem:biclique-certificate}]
Let us write $u_S'$ for the vector-valued linear function in indeterminate $y$ defined by $u_S'(i) = u_S(i) (1-y_i)$.
Then, observe that $\cB(H) \sststile{2r+2}{x,y} \Set{x_S u_S(i)y_i = x_S y_i}$.
In particular, $\cB(H) \sststile{4r}{x,y} \Set{\Norm{x_S u_S'}_2^2 \leq x_S(n-|y|)}$.
Further, for any $r \in \N$ and any $S \subseteq U$ such that $|S|=r$, we have:
\begin{equation*}
\cB(H) \sststile{4r}{x,y} \Set{x_S \sum_{i} u_S'(i) = x_S \sum_{i} u_S(i)(1-y_i) \leq  \Delta_{2r} x_S - x_S |y|}\mper
\end{equation*}

Next, let $S,T \subseteq U$ such that $S \neq T$ and $|S|,|T|\leq r$. 
Then, by noting that $u_{S}' \circ u_{T}' = u_{S \Delta T}'$, we have:
\begin{equation*}
\cB(H) \sststile{4r}{x,y} \Set{x_S x_T \iprod{u_S',u_T'} = x_{S \cup T} \sum_{i} u_{S \Delta T}(i)(1-y_i) \leq  \Delta_{2r} x_{S \cup T}  - x_{S \cup T} |y|}\mper
\end{equation*}

Next, we have:~\footnote{This is a sum-of-squares proof of the classical fact upper bounding the number of negatively correlated vectors in $n$ dimensions.}

\begin{align*}
\cB(H) &\sststile{4r}{x,y} \Biggl\{ 0 \leq \Norm{\sum_{|S| = r} x_S u_S'}_2^2 = \sum_{|S| = r}  \Norm{x_S u_S'}_2^2 + \sum_{S \neq T} \Iprod{x_S u_{S}', x_T u_T'}\\
&\leq \Paren{\sum_{|S|=r} x_S} \Paren{n-|y|} + \Delta_{2r} \sum_{S \neq T} \Paren{ x_{S \cup T}   - x_{S \cup T} |y|}\\
&\leq  n \Paren{\sum_{|S|=r} x_S} + \Delta_{2r} \sum_{S,T \subseteq U, |S|,|T|=r }  x_{S \cup T}   - \sum_{S,T \subseteq U, |S|,|T|=r } x_{S \cup T} |y|\\
&= \Paren{\sum_{|S|=r} x_S} n +  \Delta_{2r} \Paren{\sum_{|S|=r} x_S}^{2}   - \Paren{\sum_{|S|=r} x_S}^2 |y|
\Biggr\} \mper
\end{align*}

Rearranging gives:
\begin{align*}
\cB(H) &\sststile{4r}{x,y} \Biggl\{  \Paren{\sum_{|S|=r} x_S}^{2} |y| \leq 
 n \Paren{\sum_{|S|=r} x_S} + \Delta_{2r} \Paren{\sum_{|S|=r} x_S}^{2} 
\Biggr\}\mper
\end{align*}

\end{proof}

\subsection{The case of arbitrary $p$}
\label{subsec:cert-general}

In this section, we generalize the certificates of Section~\ref{sec:cert-half} to general edge densities.
The certificates use the same system of polynomial constraints $\cB(H)$ as in the previous section.

\begin{theorem}[Sum-of-squares certificates for unbalanced bicliques in random bipartite graphs for general densities] \label{thm:certification-large-p}
  Let $H \sim B(k,n-k, p)$ be a bipartite \Erdos-\Renyi with edge probability $p$.
  Then with probability $0.99$ we obtain the following two bounds:
  \begin{enumerate}
    \item Fix any $\epsilon > 0$ independently of the other parameters. For $p, 1-p \geq n^{-(1-\epsilon)}$,
    \begin{equation*}
      \cB(H) \sststile{4}{x,y} \Set{ |x| |y| \leq 
      O\Paren{\frac{np}{1-p}} }\mper 
      \end{equation*}
  \item For any $r$ such that $k \geq \max\{O(\sqrt{r n \log n}p^{2r}/(1-p)^{2r+1}), O((\log n) p/(1-p))\}$,
  \begin{equation*}
    \cB(H) \sststile{4r+2}{x,y} \Set{ |x|^{4r}|y| \leq 
       (1000r)^{10r} n \Paren{\frac{n\max\{p/(1-p),(1-p)/p\}^rp^r}{k(1-p)^{r+1}}}^4 }\mper 
    \end{equation*}
  \end{enumerate} 

  Further, the total bit complexity of the sum-of-squares proofs is $n^{O(1)}$ and $n^{O(r)}$, respectively.
\end{theorem}

For the proof of Theorem~\ref{thm:certification-large-p}, we will work with matrices with $p$-biased characters as entries.
We first define these well-studied objects. 

\begin{definition}[$p$-biased characters and normalized adjacency matrix] \label{def:char-matrix}
  Let $H = (U, V, E)$ be a bipartite graph.
  We define the $p$-biased character corresponding to an edge $H(i, j)$ to be
  \[H_p(i, j) = \begin{cases}
    \sqrt{\frac{1-p}{p}} & \text{if } H(i, j) = 1\mcom\\
    -\sqrt{\frac{p}{1-p}} & \text{if } H(i, j) = -1\mper
  \end{cases}\]
  The normalized adjacency matrix $H_p$ of the graph is matrix with the $(i,j)$-th entry equal to $H_p(i,j)$. 
\end{definition}

Let us first analyze a simple spectral certificate (that confirms that our algorithm recovers the bounds of~\cite{MMT20,DBLP:conf/stoc/CharikarSV17} from a basic relaxation in our scheme) in order to recover the first bound above. 

\begin{lemma}[Simple spectral certificate] \label{lem:simple-spectral-certificate}
Let $H \sim B(k,n-k, p)$ be a bipartite \Erdos-\Renyi with edge probability $p$. Fix any $\epsilon > 0$. Then, for any $p,1-p \geq n^{-(1-\epsilon)}$, we have:
\begin{equation*}
\cB(H) \sststile{4}{x,y} \left\{ |x| |y| = \Norm{H_p}_2^2 \leq O\Paren{\frac{np}{1-p}} \right\}\mper
\end{equation*}
Further, the total bit complexity of the sum-of-squares proof is $n^{O(1)}$.
\end{lemma}
\begin{proof}
We have:
\begin{align*}
\cB(H) \sststile{4}{x,y} \left\{ \frac{1-p}{p} |x|^2 |y|^2 = \Paren{x^{\top} H_p y}^2 \leq  \Norm{x}_2^2 \Norm{H_py}_2^2 \leq \Norm{x}_2^2 \Norm{H_p}_2^2 \Norm{y}_2^2 \right\}\mper
\end{align*}
In the inequality above, we used the sum-of-squares Cauchy-Schwarz inequality. 

Applying the first part of Lemma~\ref{lem:simple-cancel} with $a = \Norm{x}_2^2 \Norm{y}_2^2$:
\begin{align*}
\cB(H) \sststile{4}{x,y} \left\{ |x|^2|y|^2 \leq \frac{p^2}{(1-p)^2} \Norm{H_p}_2^4 \right\}\mper
\end{align*}
Applying the second part of Lemma~\ref{lem:simple-cancel} with $a = \Norm{x}_2^2 \Norm{y}_2^2$ we obtain:
\begin{align*}
\cB(H) \sststile{4}{x,y} \left\{ |x| |y| \leq \frac{p}{1-p}\Norm{H_p}_2^2 \right\}\mper
\end{align*}

Finally, notice that the entries of $H_p$ are mean $0$, variance $1$ and are bounded above by $\max\{\sqrt{p/1-p}, \sqrt{1-p/p}\}$ in magnitude. 
For $p, 1-p \geq n^{-(1-\epsilon)}$, the bound on the entries evaluates to $n^{0.5-\epsilon/2}$.
So we can apply Fact~\ref{fact:char-matrix-spectral-norm} to conclude that $\Norm{H_p}_2 =O(\sqrt{n})$ with probability at least $0.99$.

\end{proof}

The proof of second bound in Theorem~\ref{thm:certification-large-p} uses a generalization of $r$-fold balancedness defined in terms of $p$-biased characters.
We call this new property $r$-fold $p$-balancedness.

\begin{definition}[Balancednes for general densities]
  Let $H = (U, V, E)$ be a bipartite graph.
  For every $S \subseteq U$, let $u_{p,S}$ be the $|V|$-dimensional vector defined by setting $u_{p,S}(j)= \prod_{i \in S} H_p(i,j)$.
  Then, we say that $H$ has \emph{$2r$-fold $p$-balancedness} $\Delta$ if, for all $S, T \subseteq U$ of size $|S|,|T| \leq r$, it holds that $|\sum_{j \in V} u_{p,S}(j) u_{p,T}(j)| \leq \Delta$.
\end{definition}

The following lemma verifies $r$-fold $p$-balancedness of random bipartite graphs, as well as an upper bound on the maximum degree of the vertices on the righ-hand side.

\begin{lemma}[Balancedness of random bipartite graphs for general densities]
\label{lem:balacedness-density-large-p}
Let $H = (U,V,E) \sim B(k, n-k, p)$.
Then, for any $r \leq k$, with probability at least $0.99$ over the draw of $H$,
1) $H$ has $2r$-fold $p$-balancedness $O(\sqrt{r n \log k}p^r/(1-p)^r)$, and
2) the maximum degree of a vertex in $V$ is at most $kp+O(\sqrt{kp(1-p) \log n})$.
\end{lemma}
\begin{proof}
For $S, T \subseteq U$ such that $|S|,|T| \leq r$, we have that $u_{p,S}(j) u_{p,T}(j)$ has mean $0$ and is bounded between $-\Paren{\sqrt{\frac{p}{1-p}}}^{2r} = -p^r/(1-p)^r$ and $\Paren{\sqrt{\frac{p}{1-p}}}^{2r} = p^r/(1-p)^r$.
Then, by Hoeffding's inequality, 
\[\Pr\Brac{ \left| \sum_{j \in V} u_{p,S}(j)u_{p,T}(j) \right| \geq t \sqrt{|V|} p^r / (1-p)^r} \leq 2e^{-t^2/2} \mcom \]
so, by a union bound over all choices of $S$ and $T$,
\[\Pr\Brac{ \exists S, T \subseteq U \text{ s.t. } |S|, |T| \leq r, \left| \sum_{j \in V} u_{p,S}(j)u_{p,T}(j) \right| \geq t \sqrt{|V|} p^r / (1-p)^r} \leq |U|^{2r} \cdot 2e^{-t^2/2} \mper \]
Choosing $t = O(\sqrt{r \log U})$ makes the right-hand side a small constant, so we have $2r$-fold $p$-balancedness $O(\sqrt{r|V| \log |U|}p^r/(1-p)^r)$.

The degree of a vertex in $V$ is a binomial random variable $\operatorname{Bin}(k, p)$, which by standard bounds is larger than $kp+t$ with probability at most $\min\{e^{-t^2/(2k(1-p))}, e^{-t^2/(2kp+2t/3)}\}$.
By a union bound over all vertices in $V$, the maximum degree is larger than $kp+t$ with probability at most $|V|\min\{e^{-t^2/(2k(1-p))}, e^{-t^2/(2kp+2t/3)}\}$, so choosing $t=O(\sqrt{kp(1-p) \log |V|})$ makes the probability a small constant.
Hence, the maximum degree is $kp+O(\sqrt{kp(1-p)\log |V|})$.
\end{proof}

The following lemma is the key component of the proof of Theorem~\ref{thm:certification-large-p}, and is analogous to Lemma~\ref{lem:biclique-certificate} in Section~\ref{sec:cert-half}.

\begin{lemma}\label{lem:biclique-certificate-large-p}
  Let $H = (U, V, E)$ be a bipartite graph with $|U|=k$ and $|V|=n-k$ and $2r$-fold $p$-balancedness $\Delta_{2r}$.
  Then, 
  \begin{equation}
  \cB(H) \sststile{4r}{x,y} \Biggl\{ \Paren{1-p}^{r}/p^r \Paren{\sum_{|S|=r} x_S}^{2} \sum_i y_i  \leq 
   n \max\{p/(1-p),(1-p)/p\}^r \Paren{\sum_{|S|=r} x_S} + \Delta_{2r} \Paren{\sum_{|S|=r} x_S}^{2}
  \Biggr\}\mper \label{eq:biclique-main-2-large-p}
  \end{equation} Further, the total bit complexity of the sum-of-squares proof is $n^{O(r)}$.
\end{lemma}

We postpone the proof of the lemma, and combine the result with an observation analogous to that in Lemma~\ref{lem:lower-bound-LHS-trick}.

\begin{lemma}[Lower bounding the LHS of \eqref{eq:biclique-main-2-large-p}]
  \label{lem:lower-bound-LHS-trick-large-p}
  Let $H = (U, V, E)$ be a bipartite graph with $|U|=k$ and $|V|=n-k$ and maximum degree of a vertex in $V$ at most $kp+\Delta_{\ell}$.
  Then, for any $j \in V$, we have:

  \begin{equation*}
  \cB(H) \sststile{4r}{x,y} \Biggl\{ y_j \Paren{\sum_{S: |S|=r} x_S}^2 |y| \geq \Paren{k(1-p)-\Delta_{\ell}} y_j \Paren{\sum_{S: |S|=r} x_S}^2  \Biggr\}\mper
  \end{equation*}
  Further, the total bit complexity of the sum-of-squares proof is $n^{O(r)}$.

\end{lemma}

\begin{proof}
Every $j \in V$ has degree at most $kp+\Delta_{\ell}$.
Thus, we have using the constraint system $\cB(H)$:
\[
\cB(H) \sststile{2}{x,y} \Set{ y_j |x| = \sum_{u \in U:u \sim j} x_u y_j \leq (kp+\Delta_\ell) y_j}
\]

Thus, $\cB(H) \sststile{2}{x,y} \Set{|x| = k-|y|}$ allows us to conclude:
\[
\cB(H) \sststile{2}{x,y} \Set{y_j |y| \geq (k(1-p)-\Delta_{\ell}) y_j}
\]
Multiplying both sides by the sum-of-squares polynomial $\Paren{\sum_{S: |S|=r} x_S}^2$ completes the proof.
\end{proof}

As a direct consequence of Lemma~\ref{lem:biclique-certificate-large-p} and Lemma~\ref{lem:lower-bound-LHS-trick-large-p}, we obtain:

\begin{lemma}\label{lem:biclique-certificate-final-large-p}
  Let $H = (U, V, E)$ be a bipartite graph with $|U|=k$ and $|V|=n-k$ and $2r$-fold $p$-balancedness $\Delta_{2r}$ and maximum degree of a vertex in $V$ at most $kp+\Delta_{\ell}$.
Suppose further that
\begin{equation*}
k(1-p)^{r+1}/p^r - \Delta_\ell (1-p)^{r}/p^r - \Delta_{2r} \geq \frac{k}{2}(1-p)^{r+1}/p^r\mper
\end{equation*}
Then, we have
\begin{equation*}
\cB(H) \sststile{4r+2}{x,y} \Set{ \Paren{\sum_{S: |S|=r}x_S}^4 |y| \leq n \Paren{\frac{2n\max\{p/(1-p),(1-p)/p\}^rp^r}{k(1-p)^{r+1}}}^4
}\mper 
\end{equation*}
Further, the total bit complexity of the sum-of-squares proof is $n^{O(r)}$.
\end{lemma}

\begin{proof}
We first multiply both sides of the conclusion of Lemma~\ref{lem:biclique-certificate-large-p} with the sum-of-squares polynomial $y_j^2$ for an arbitrary $j \in V$:

\begin{align*}
\cB(H)
&\sststile{4r+2}{x,y} \Biggl\{
  (1-p)^{r}/p^r y_j  \Paren{\sum_{|S|=r} x_S}^{2} |y|\\
&\leq 
 n \max\{p/(1-p),(1-p)/p\}^r y_j \Paren{\sum_{|S|=r} x_S}  + \Delta_{2r} y_j \Paren{\sum_{|S|=r} x_S}^{2} \Biggr\}\mper
\end{align*}

Next, we use Lemma~\ref{lem:lower-bound-LHS-trick-large-p} to replace the left-hand side in the above by a useful lower bound:

\begin{align*}
\cB(H)
&\sststile{4r+2}{x,y} \Biggl\{
  \Paren{k(1-p)-\Delta_{\ell}} (1-p)^r/p^r y_j \Paren{\sum_{S: |S|=r} x_S}^2 \\
&\leq 
n \max\{p/(1-p),(1-p)/p\}^r y_j \Paren{\sum_{|S|=r} x_S} + \Delta_{2r} y_j  \Paren{\sum_{|S|=r} x_S}^{2}\Biggr\}\mper
\end{align*}

We then move the second term on the right-hand side to the left-hand side and use that $k(1-p)^{r+1}/p^r - \Delta_\ell (1-p)^{r}/p^r - \Delta_{2r} \geq \frac{k}{2}(1-p)^{r+1}/p^r$ to conclude:

\begin{equation*}
\cB(H) \sststile{4r+2}{x,y} \Set{y_j \Paren{\sum_{S: |S|=r} x_S}^2
\leq \frac{2n\max\{p/(1-p),(1-p)/p\}^rp^r}{k(1-p)^{r+1}} y_j \Paren{\sum_{|S|=r} x_S}}\mper
\end{equation*}
We finally apply Lemma~\ref{fact:cancellation-within-SoS} with $a = y_j (\sum_{S: |S|=r} x_S)$, $C=\frac{2n\max\{p/(1-p),(1-p)/p\}^rp^r}{k(1-p)^{r+1}}$, and $t=2$ to obtain:
\begin{equation*}
\cB(H)
\sststile{4r+2}{x,y} \Set{
y_j \Paren{\sum_{S: |S|=r}x_S}^4 \leq \Paren{\frac{2n\max\{p/(1-p),(1-p)/p\}^rp^r}{k(1-p)^{r+1}}}^4}\mper
\end{equation*}

Summing up as $j$ varies over $V$ completes the proof.

\end{proof}

We now finish the proof of Theorem~\ref{thm:certification-large-p}:

\begin{proof}[Proof of Theorem~\ref{thm:certification-large-p}]

The first bound follow by Lemma~\ref{lem:simple-spectral-certificate}.
In the rest of the proof we focus on the second bound.

From Lemma~\ref{lem:power-vs-sets}, we have:
\begin{equation*}
\Set{ x_i^2 = x_i \text{ } \forall i \in U } \sststile{r}{x} \Set{\frac{1}{2^r r!} |x|^r - \frac{2r^r}{r!} - \epsilon \leq \sum_{|S|=r} x_S}\mper
\end{equation*}

Setting $\epsilon =1$ and using that $\{0 \leq a \leq C\} \sststile{4}{a,C} \{a^4 \leq C^4\}$, we have:
\begin{equation*}
\Set{ x_i^2 = x_i \text{ } \forall i \in U } \sststile{4r}{x} \Set{|x|^{4r} \leq (100r)^{10r} + (100r)^{10r} \Paren{\sum_{|S|=r} x_S}^{4} }\mper
\end{equation*}

Now we want to combine this with the conclusion of Lemma~\ref{lem:biclique-certificate-final-large-p}.
We briefly verify that we satisfy the condition $k(1-p)^{r+1}/p^r - \Delta_\ell (1-p)^{r}/p^r - \Delta_{2r} \geq \frac{k}{2}(1-p)^{r+1}/p^r$. 
We have by Lemma~\ref{lem:balacedness-density-large-p} that $\Delta_{\ell}=O(\sqrt{kp(1-p)\log|V|}) = O(\sqrt{k(1-p)\log n})$ and $\Delta_{2r}=O(\sqrt{r|V|\log|U|}p^r/(1-p)^r) = O(\sqrt{rn\log n}p^r/(1-p)^r)$.
Observe that for $k \geq \max\{O((\log n) p/(1-p)), O(\sqrt{r n \log n}p^{2r}/(1-p)^{2r+1})\}$ large enough the condition is satisfied.
Then we have:
\begin{equation*}
\cB(H) \sststile{O(r)}{x,y} \Set{ |x|^{4r} |y| \leq 
(100r)^{10r} |y| + (100r)^{10r} n \Paren{\frac{2n\max\{p/(1-p),(1-p)/p\}^rp^r}{k(1-p)^{r+1}}}^4 }\mper 
\end{equation*}

Observing that $\cB(H) \sststile{2}{y} \Set{ |y| \leq n}$ completes the proof.

\end{proof}

Finally, we complete the proof of Lemma~\ref{lem:biclique-certificate-large-p}.

\begin{proof}[Proof of Lemma~\ref{lem:biclique-certificate-large-p}]
Let us write $u_{p,S}'$ for the vector-valued linear function in indeterminate $y$ defined by $u_{p,S}'(i) = u_{p,S}(i) (1-y_i)$.
Then, observe that $\cB(H) \sststile{4r}{x,y} \Set{x_{S \cup T} u_{p,S}(i) u_{p,T}(i) y_i = x_{S\cup T} y_i (1-p)^r/p^r}$ and
\begin{align*}
\cB(H)
&\sststile{4r}{x,y} \Biggl\{\Norm{x_S u_{p,S}'}_2^2 = \sum_i x_S u_{p,S}(i)^2 (1-y_i)\\
&\leq n \max\{p/(1-p),(1-p)/p\}^r x_S - (1-p)^r/p^r x_S |y| \Biggr\}\mper
\end{align*}

Let $S,T \subseteq U$ such that $S \neq T$ and $|S|,|T|\leq r$. We have:
\begin{align*}
\cB(H)
&\sststile{4r}{x,y} \Biggl\{x_S x_T \iprod{u_S',u_T'} = x_{S \cup T} \sum_{i} u_{S}(i)u_T(i)(1-y_i)
\leq  \Delta_{2r} x_{S \cup T}   - \Paren{1-p}^{r}/p^r x_{S \cup T} |y| \Biggr\}\mper
\end{align*}

Then, we have:

\begin{align*}
\cB(H) &\sststile{4r}{x,y} \Biggl\{ 0 \leq \Norm{\sum_{|S| = r} x_S u_S'}_2^2 = \sum_{|S| = r}  \Norm{x_S u_S'}_2^2 + \sum_{S \neq T} \Iprod{x_S u_{S}', x_T u_T'}\\
&\leq \Paren{\sum_{|S|=r} x_S} \Paren{n \max\{p/(1-p),(1-p)/p\}^r -  \Paren{1-p}^{r}/p^r |y|}\\
&\quad + \Delta_{2r} \sum_{S \neq T} \Paren{  x_{S \cup T}  - \Paren{1-p}^{r}/p^r} x_{S \cup T} |y| \\
&\leq n \max\{p/(1-p),(1-p)/p\}^r \Paren{\sum_{|S|=r} x_S} + \sum_{S,T \subseteq U, |S|,|T|=r}  \Delta_{2r} x_{S \cup T} \\
&\quad - \Paren{1-p}^{r}/p^r \sum_{S,T \subseteq U, |S|,|T|=r} x_{S \cup T} |y| \\
&= n \max\{p/(1-p),(1-p)/p\}^r \Paren{\sum_{|S|=r} x_S} + \Delta_{2r} \Paren{\sum_{|S|=r} x_S}^2  - \Paren{1-p}^{r}/p^r \Paren{\sum_{|S|=r} x_S}^2 |y| 
\Biggr\} \mper
\end{align*}

Rearranging gives:
\begin{align*}
\cB(H) &\sststile{4r}{x,y} \Biggl\{ \Paren{1-p}^{r}/p^r \Paren{\sum_{|S|=r} x_S}^{2} |y|  \leq 
n \max\{p/(1-p),(1-p)/p\}^r \Paren{\sum_{|S|=r} x_S} + \Delta_{2r} \Paren{\sum_{|S|=r} x_S}^{2} 
\Biggr\}\mper
\end{align*}

\end{proof}

\section{List-decoding semi-random planted cliques}
In this section, we describe our algorithm for list-decoding semi-random planted cliques using high-constant degree sum-of-squares relaxations.
We will abstract out our requirement of sum-of-squares refutation of biclique numbers in random bipartite graphs in order to transparently show that the explicitness of the certificate is irrelvant to our algorithm.
In Section~\ref{subsec:main-results}, we will immediately obtain our algorithmic results as a direct consequence of our certificates from the previous section and an elementary cleanup step that takes a list with an approximately correct candidate and fixes it up to a list containing the planted clique $S^*$. 

\begin{theorem} \label{thm:rounding-from-biclique-certs}
Fix any $t \in \N$. 
There is an $n^{O(t)}$ time algorithm that takes as input a graph $G$ on $n$ vertices with the following guarantees.
Suppose $G$ has a clique $S^*$ of size $k$ in it. 
Suppose that the bipartite graph $H$ defined by keeping only the edges from $\cut(S^*)$ in $G$ admits an $O(t)$-th order sum-of-squares certificate of unbalanced biclique number as below for some function $\omega=\omega(n,k,p)$: \pravesh{we should use a different letter here since $\omega =\omega(k,n,p)$ stands for the bipartite graph model.}
\[
\cB(H) \sststile{O(t)}{x,y} \Set{|x|^{t} |y| \leq \omega}\mper
\]
Then, if $\omega \cdot (n/k^2)^t \leq \delta k$, the algorithm outputs a list of $O((n/k)^t)$ subsets, each of size at most $k/(1-2\delta)$ such that with probability at least $0.99$ over the randomness of the algorithm there is an element $S$ of the list that satisfies $|S \cap S^*| \geq (1-2\delta) k$.
\end{theorem}

In the main results, the list that comes from Theorem~\ref{thm:rounding-from-biclique-certs} will be pruned (using that $S^*$ has small intersection with other $k$-cliques, see Lemma~\ref{lem:clique-intersection}) and refined to consist of $(1+o(1))n/k$ cliques of size $k$.

We will prove Theorem~\ref{thm:rounding-from-biclique-certs} using the following natural algorithm.
Recall the standard $k$-clique constraint system $\cA$ defined earlier. 
Our rounding scheme is reminiscent of those used in rounding algorithms for list-decodable learning~\cite{DBLP:conf/nips/KarmalkarKK19,DBLP:conf/soda/BakshiK21,DBLP:conf/stoc/IvkovK22}.

\begin{mdframed}
  \begin{algorithm}[List-decoding semi-random planted cliques]
    \label[algorithm]{alg:list-decoding-semirandom-planted-clique}\mbox{}
    \begin{description}
    \item[Given:] A graph $G$ on $n$ vertices with a clique $S^*$ of size $k$.
    \item[Output:]
      A list $L \subseteq \R^d$ of size $O((n/k)^t)$ that contains an $S$ such that $|S \cap S^*| \geq (1-\delta)k$.
    \item[Operation:]\mbox{}
    \begin{enumerate}
    \item Find a degree-$O(t)$ pseudo-distribution $D$ on $w$ satisfying the $k$-clique axioms on $\cA(G)$ and minimizing $\|\pE_{D}[w]\|_2$.
    \item For every $Q \in [n]^t$, an ordered $t$-tuple on $[n]$ such that $\pE_D[w_Q]>0$, let $C_Q = \frac{\pE_{D}[w_Q w]}{\pE_{D}[w_Q]}$.    
    \item For $N= O((n/k)^t)$ repetitions, choose an ordered $t$-tuple $Q \in [n]^t$ with probability proportional to $\pE_{D}[ w_Q]$  and add $C_Q$ to the list $\cL'$. 
    \item For each element $C_Q \in \cL'$, construct the set $S_Q = \{i \mid C_Q(i) \geq 1-2\delta\}$ and add it to $\cL$.
    \item Output $\cL$.
  \end{enumerate}
    \end{description}    
  \end{algorithm}
\end{mdframed}

To analyze this algorithm, we first observe that the maximal coverage property (i.e., $D$ minimizing $\Norm{\pE_D[w]}_2$) implies that $\pE_{D}[w]$ has a non-trivial weight on the true (but unknown) $k$-clique $S^*$. 
This lemma is by now standard with analogous usages in the context of list-decodable learning~\cite{DBLP:conf/nips/KarmalkarKK19,DBLP:conf/soda/BakshiK21,DBLP:conf/stoc/IvkovK22}.
It can be proven by showing that if $\sum_{i \in S^*} \pE_D[w_i] < k^2/n$ then one can take a ``mix'' of $D$ and the distribution that places all its mass on $S^*$ (which does satisfy $\cA$) and produce another pseudo-distribution $D'$ with smaller $\Norm{\pE_D[w_i]}_2$. 

\begin{lemma}[Maximal coverage implies non-trivial weight on $S^*$, see Lemma 4.3 in~\cite{DBLP:conf/nips/KarmalkarKK19}]
\label{lem:max-coverage}
Let $D$ be a pseudo-distribution of degree $\geq 4$ satisfying $\cA(G)$ that minimizes $\Norm{\pE_D[w]}_2$.
Then, $\sum_{i \in S^*} \pE_D[w_i] \geq k^2/n$.
\end{lemma}

As an immediate corollary, we observe the following consequence of our rounding scheme:

\begin{lemma}
\label{lem:choosing-within-S-star}
Let $D$ be the pseudo-distribution constructed in Step 1 of the algorithm. 
Then, in Step 3 of the algorithm, each of the chosen $t$-tuples $Q$ satisfies $Q \in (S^*)^t$ with probability at least $(k/n)^t$. 
\end{lemma}
\begin{proof}
The probability that in Step 3 of the algorithm an ordered $t$-tuple $Q$ is in $(S^*)^t$ is $\pE_D[\Paren{\sum_{i \in S^*} w_i}^{t}]/k^t$, where we used that $\pE_D[\Paren{\sum_{i=1}^n w_i}^t]=k^t$.
The proof now follows by applying Hölder's inequality for pseudo-distributions to conclude that $\pE_D[\Paren{\sum_{i \in S^*} w_i}^t] \geq \pE_D [ \sum_{i \in S^*} w_i]^t \geq (k^2/n)^t$ (from Lemma~\ref{lem:max-coverage}).
\end{proof}

Next, we argue that for $Q \in (S^*)^t$ chosen with probability proportional to $\pE_{D}[ w_Q]$, with probability at least $0.5$, the corresponding $S=C_Q$ has a non-trivial intersection with $S^*$. 

\begin{lemma}
\label{lem:good-vector-conditioned-Q-in-S-star}
Assume the hypothesis of Theorem~\ref{thm:rounding-from-biclique-certs}. 
Then, in Step 4 of the algorithm, conditioned on $Q \in (S^*)^t$, with probability at least $0.5$, $\sum_{i \in S^*} C_Q(i) \geq (1-\delta)k$.
\end{lemma}

\begin{proof}
Consider the bipartite graph $H$ formed by keeping only the edges that lie in $\cut(S^*)$ in $G$ with left vertex set equal to $S^*$ and right vertex set equal to $[n] \setminus S^*$.
Then, observe that $\cA(G) \sststile{}{} \cB(H)$ via the polynomial map $x_u = w_u$ for every $u \in S^*$ and $y_v = w_v$ for every $v \in [n] \setminus S^*$.
Since $D$ satisfies $\cA(G)$, the polynomial transformation above applied to $D$ gives a pseudo-distribution on $(x,y)$ that satisfies $\cB(H)$. 
From the biclique certificate, we have: $\pE_D[ |x|^t|y|] \leq \omega$. 

This yields that 
\[
\sum_{i_1, i_2,\ldots, i_t} \pE_{D}[x_{i_1} x_{i_2} \cdots x_{i_t} |y|] \leq \omega \mper
\]

Rescaling and rewriting yields
\[
\frac{1}{\pE_D [|x|^t]} \sum_{i_1, i_2,\ldots, i_t: \pE_{D}[x_{i_1} x_{i_2} \cdots x_{i_t}] >0} \pE_{D}[x_{i_1} x_{i_2} \cdots x_{i_t}] \frac{\pE_{D}[x_{i_1} x_{i_2} \cdots x_{i_t}|y|]}{\pE_{D}[x_{i_1} x_{i_2} \cdots x_{i_t}]} \leq \frac{\omega}{\pE_D [|x|^t]} \mper
\]
Using Hölder's inequality for pseudo-distributions and Lemma~\ref{lem:max-coverage}, we know that $\pE_D [|x|^t] = \pE_D [(\sum_{i \in S^*} w_i)^t] \geq \pE_D [\sum_{i \in S^*} w_i]^t \geq (k^2/n)^t$.
Thus the right-hand side of the above is at most $\omega (n/k^2)^t$. 

Observe that the left-hand side can be interpreted as the expected value of the random variable $\frac{\pE_{D}[x_{i_1} x_{i_2} \cdots x_{i_t}|y|]}{\pE_{D}[x_{i_1} x_{i_2} \cdots x_{i_t}]}$ where each $i_1,i_2, \ldots, i_t$ is chosen with probability equal to $\frac{\pE_{D}[x_{i_1} x_{i_2} \cdots x_{i_t}]}{\sum_{i_1, i_2, \ldots, i_t}\pE_{D}[x_{i_1} x_{i_2} \cdots x_{i_t}]}$. 

For an ordered tuple $Q \in (S^*)^t$ chosen with probability proportional to $\pE_{D}[w_Q]$, consider the $(n-k)$-dimensional vector $\frac{\pE_{D}[w_Q y]}{\pE_{D}[w_Q]}$.
Its $\ell_1$-norm is equal to $\frac{\pE_{D}[w_Q |y|]}{\pE_{D}[w_Q]}$, which is also equal to the random variable $\frac{\pE_{D}[x_{i_1} x_{i_2} \cdots x_{i_t}|y|]}{\pE_{D}[x_{i_1} x_{i_2} \cdots x_{i_t}]}$ where each $i_1,i_2, \ldots, i_t$ is chosen with probability equal to $\frac{\pE_{D}[x_{i_1} x_{i_2} \cdots x_{i_t}]}{\sum_{i_1, i_2, \ldots, i_t}\pE_{D}[x_{i_1} x_{i_2} \cdots x_{i_t}]}$.

Thus, we have concluded that the expected value of the $\ell_1$-norm of $\frac{\pE_{D}[w_Q y]}{\pE_{D}[w_Q]}$ is at most $\omega (n/k^2)^t$.
By Markov's inequality, with probability at least $0.5$ over the choice of $Q$, thus, the $\ell_1$-norm of $\frac{\pE_{D}[w_Q y]}{\pE_{D}[w_Q]}$ is at most $2\omega (n/k^2)^t$. Also note that, by Fact~\ref{fact:reweightings}, $\sum_{i=1}^n C_Q(i) \geq k$. Then, with probability at least $0.5$ over the choice of $Q$, $\sum_{i \in S^*} C_Q(i) \geq k - \Norm{\frac{\pE_{D}[w_Q y]}{\pE_{D}[w_Q]}}_1 \geq k- \omega (n/k^2)^t \geq (1-\delta)k$ using that $\omega (n/k^2)^t \leq \delta k$.

\end{proof}

\begin{proof}[Proof of Theorem~\ref{thm:rounding-from-biclique-certs}]

From Lemma~\ref{lem:choosing-within-S-star}, in Step 4, we choose a $Q \subseteq S^*$ with probability at least $(k/n)^t$.
Conditioned on this event happening, Lemma~\ref{lem:good-vector-conditioned-Q-in-S-star} shows that $\sum_{i \in S^*}C_Q(i) \geq (1-\delta)k$ with probability at least $0.5$.
We call such $Q$ good.

By averaging, for a good $Q$, we must have that for a $(1-2\delta)$-fraction of $i \in S^*$, $C_Q(i) \geq 1-2\delta$.
Further, the total number of coordinates of $C_Q$ larger than $1-2 \delta$ cannot be more than $k/(1-2\delta)$.
Thus, $S_Q$ is a set of size at most $k/(1-2\delta)$ such that $|S_Q \cap S^*| \geq (1-2\delta) k$.

The $O((n/k)^t)$ repetitions in Step 3 ensure that with probability at least $0.99$ we choose at least one good $Q$.
\end{proof}

\subsection{Proof of main results} \label{subsec:main-results}

We combine our biclique certificates, the rounding algorithm, and a simple cleanup step to obtain the main results of our work.
We start by stating the main result.

\begin{theorem}[Main result]
  \label{thm:main-result}
  Consider a graph $G$ on $n$ vertices such that $G$ is generated according to $\FK(n,k,p)$.
  Then the following two results hold:
  \begin{enumerate}
  \item \emph{($n^{2/3}$ guarantee)} For any $\epsilon>0$ and $p,1-p\geq n^{-(1-\epsilon)}$, there exists an algorithm that takes input $G$, runs in polynomial time, and for $k \geq \max\{O(n^{2/3}p^{1/3}/(1-p)^{2/3}), \tilde{O}(n^{1/2})\}$, with probability $0.99$ outputs a list of at most $(1+o(1))n/k$ $k$-cliques such that one of them is the planted clique in $G$.
  
  \item \emph{($n^{1/2+\epsilon}$ guarantee)} For any $\epsilon>0$ small enough, there exists an algorithm that takes input $G$, runs in time $n^{O(1/\epsilon)}$, and for $k \geq n^{1/2 + \epsilon}/(1-p)^{1/\epsilon}$, with probability $0.99$ outputs a list of at most $(1+o(1))n/k$ $k$-cliques such that one of them is the planted clique in $G$.
  \end{enumerate}
  \end{theorem}

Before we prove this, we state and prove two auxiliary lemmas that help us prune the list of subsets returned by the list-decoding algorithm in Theorem~\ref{thm:rounding-from-biclique-certs}.

\begin{lemma}[Intersection of cliques with the planted clique]\label{lem:clique-intersection}
Let $G \sim \FK(n,k,p)$. Let $S^*$ be the planted clique in $G$. Then, with probability at least $1-\frac{k}{n^2}$, any other clique $S$ of size at least $k$ satisfies $|S\cap S^*| \leq 3 \frac{\log n}{\log 1/p}$.
\end{lemma}
\begin{proof}
The proof is analogous to that of Proposition~\ref{prop:biclique-number-overview} and is an easy consequence of a Chernoff bound and a union bound.
\end{proof}

\begin{lemma}[Subsets with small intersection]
\label{lem:sets-pinex}
Let $S_1, ..., S_m \subseteq [n]$ with $|S_i| = k$ and $|S_i \cap S_j| \leq \Delta$.
Then, if $k \geq \sqrt{2n\Delta}$, we have $m \leq \frac{n}{k} \Paren{1+\frac{2n\Delta}{k^2}}$.
\end{lemma}
\begin{proof}
By the inclusion-exclusion principle, we need
\[mk - \frac{m^2}{2}\Delta \leq n\mper\]
By inspecting the above as a quadratic equation in $m$, we get that for $k \geq \sqrt{2n\Delta}$ the equation is violated when $m > \frac{k-\sqrt{k^2-2n\Delta}}{\Delta}$.
We note that 
\[\frac{k-\sqrt{k^2-2n\Delta}}{\Delta} = \frac{k}{\Delta}\Paren{1-\sqrt{1-\frac{2n\Delta}{k^2}}} \leq \frac{k}{\Delta} \frac{n\Delta}{k^2} \Paren{1+\frac{2n\Delta}{k^2}} = \frac{n}{k}\Paren{1+\frac{2n\Delta}{k^2}}\]
and therefore obtain that $m \leq \frac{n}{k} \Paren{1+\frac{2n\Delta}{k^2}}$.
\end{proof}

We are now ready to prove the main result. The $n^{2/3}$ guarantee uses the first certificate in Theorem~\ref{thm:certification-large-p} and produces a result similar to that of \cite{MMT20}, and the $n^{1/2+\epsilon}$ guarantee uses the second certificate in Theorem~\ref{thm:certification-large-p} and is the main contribution of our work.

We start by proving the $n^{2/3}$ guarantee.

\begin{proof}[Proof of $n^{2/3}$ guarantee in Theorem~\ref{thm:main-result}]
First, we note that $\mathcal{A}(G)$ implies $\mathcal{A}(G')$ for any $G'$ that is obtained by adding edges to $G$. Therefore, in our sum-of-squares programs we can ignore the adversarial deletion phase and assume that we work with a graph in which the edges going out from $S^*$ are random.

By the first certificate in Theorem~\ref{thm:certification-large-p}, we have
\begin{equation*}
\cB(H) \sststile{4}{x,y} \Set{ |x| |y| \leq 
O\Paren{\frac{np}{1-p}} }\mper 
\end{equation*}
Next, we want to apply Theorem~\ref{thm:rounding-from-biclique-certs} with $\omega = O(np/(1-p))$ and $\delta = (1-p)/24$. To apply the theorem, we need $\omega \cdot (n/k^2) \leq \delta k$, which we rewrite as 
\begin{align*}
k \geq O\Paren{\frac{n^{2/3}p^{1/3}}{(1-p)^{2/3}}} \mper
\end{align*}
Theorem~\ref{thm:rounding-from-biclique-certs} yields a list of $O(n/k)$ subsets, each of size at most $k/(1-(1-p)/12) \leq (1-(1-p)/6)k$, such that with probability at least $0.99$ one one them interesects the true clique $S^*$ in at least $(1-(1-p)/12)k \geq (1-(1-p)/6)k$ vertices.

To obtain a list that contains $S^*$ exactly, we will remove from each $S$ in the list all vertices that are connected to few vertices in $S$, and we will add to $S$ all vertices that are connected to many vertices in $S$. Formally, we will make use of the following claim:

\begin{claim}
With probability at least $0.99$, for all subsets $S \subseteq [n]$ with $|S \cap S^*| \geq (1-\gamma)k$ and $|S| \leq (1+\gamma) k$, every vertex $v \in S^*$ is connected to at least $(1-\gamma)k-1$ vertices in $S$ and every vertex $v \not\in S^*$ is connected to at most $kp+O(\sqrt{kp(1-p)\log n}) +2\gamma k$ vertices in $S$.
\end{claim}
\begin{proof}[Proof of claim]
The first claim is trivial: every vertex $v \in S^*$ has at least $|S \cap S^*| - 1 \geq (1-\gamma)k - 1$ edges to $S$. 

For the second claim, we begin by noting that, for a vertex $v\not\in S^*$, the number of edges to $S^*$ is at most a binomial random variable $\operatorname{Bin}(k,p)$. By standard bounds, this is larger than $kp+t$ with probability at most $\min\{e^{-t^2/(2k(1-p))}, e^{-t^2/(2kp+2t/3)}\}$.
Then, by a union bound, the probability that the number of edges is larger than $kp+t$ for any $v \not\in S^*$ is at most $n \min\{e^{-t^2/(2k(1-p))}, e^{-t^2/(2kp+2t/3)}\}$.
Choosing $t=O\Paren{\sqrt{kp(1-p)\log n}}$ makes this probability a small constant.
Then, with probability at least $0.99$, no vertex $v \not\in S^*$ has more than $kp+O\Paren{\sqrt{kp(1-p)\log n}}$ edges to $S^*$. In addition, a vertex $v \not\in S^*$ has at most $|S \setminus S^*| \leq (1+\gamma)k-(1-\gamma)k=2\gamma k$ edges to $S \setminus S^*$. Therefore, overall, it has at most $kp+O(\sqrt{kp(1-p)\log n}) +2\gamma k$ edges to $S^*$.
\end{proof}

Consider the subset $S$ in the list for which $|S\cap S^*| \geq (1-(1-p)/6)$. We can apply the claim to this subset with $\gamma=(1-p)/6$. Then, every vertex $v \in S^*$ is connected to at least $(1-(1-p)/6)k-1$ vertices in $S$, and every vertex $v \not\in S^*$ is connected to at most $(p+(1-p)/3)k+O(\sqrt{kp(1-p)\log n}) < (1-(1-p)/6)k - 1$ vertices in $S$, where we used that $k > O\Paren{(\log n)p/(1-p)}$.

Therefore, we do the following: for each subset $S$ in the list, we remove from $S$ all vertices that are connected to less than $(1-(1-p)/6)k - 1$ of the vertices in $S$, and we add to $S$ all vertices that are connected to at least $(1-(1-p)/6)k-1$ of the vertices in $S$.
This ensures that the subset $S$ for which $|S\cap S^*| \geq (1-(1-p)/6)$ is transformed by this procedure into $S^*$.

After that, we remove from the list the subsets with size different than $k$ and the subsets that are not cliques.
Then we iterate the following procedure: find $S, S'$ in the list such that $|S \cap S'| \geq O(\log n/\log 1/p)$ and remove one of them from the list.
By Lemma~\ref{lem:sets-pinex}, the resulting list has size at most $(1+o(1))n/k$, where we use that our choice of $k$ satisfies $k \geq \sqrt{2nO(\log n/\log 1/p)}$.
Furthermore, by Lemma~\ref{lem:clique-intersection}, this procedure cannot remove $S^*$ from the list, because it intersects other cliques in at most $O(\log n/\log 1/p)$ vertices.

We note that $k \geq \max\{O(n^{2/3}p^{1/3}/(1-p)^{2/3}), \tilde{O}(n^{1/2})\}$ satisfies the lower bounds on $k$ that we require.
The time complexity of the algorithm is polynomial in $n$.
\end{proof}

Finally, we prove the $n^{1/2+\epsilon}$ guarantee, which we split into the cases $p \leq 1/2$ and $p \geq 1/2$.

\begin{lemma}[$n^{1/2+\epsilon}$ guarantee of Theorem~\ref{thm:main-result}, $p \leq 1/2$]
\label{lem:result-p-small}
Fix any $\epsilon>0$ small enough.
There is an algorithm that takes input a graph $G$ on $n$ vertices, runs in time $n^{O(1/\epsilon)}$, and provides the following guarantee:
If $G$ is generated according to $\FK(n,k,p)$ with $p \leq 1/2$, for $k \geq n^{1/2 + \epsilon}$, with probability $0.99$ the algorithm outputs a list of at most $(1+o(1))n/k$ $k$-cliques such that one of them is the planted clique in $G$.
\end{lemma}
\begin{proof}
First, we note that $\mathcal{A}(G)$ implies $\mathcal{A}(G')$ for any $G'$ that is obtained by adding edges to $G$. Therefore, in our sum-of-squares programs we can ignore the adversarial deletion phase and assume that we work with a graph in which the edges going out from $S^*$ are random.

For $p < 1/2$, the second certificate in Theorem~\ref{thm:certification-large-p} is the same up to constant factors as the one in Theorem~\ref{thm:certification-density-half} for $p=1/2$.
Furthermore, for $p < 1/2$, the range of $k$ for which the second certificate in Theorem~\ref{thm:certification-large-p} holds is a superset of the range of $k$ under which the one in Theorem~\ref{thm:certification-density-half} holds.
Therefore, in this proof, we assume without loss of generality that $p=1/2$, noting that all the steps in the proof continue to be valid even if $p<1/2$.

By Theorem~\ref{thm:certification-density-half}, for $k \geq O(\sqrt{t n \log n})$, we have
\begin{equation*}
\cB(H) \sststile{4t+2}{x,y} \Biggl\{ |x|^{4t}|y| \leq 
    (1000t)^{10t}n \Paren{\frac{n}{k}}^4 \Biggr\}\mper 
\end{equation*}
Next, we want to apply Theorem~\ref{thm:rounding-from-biclique-certs} with $\omega=(1000t)^{10t}n(n/k)^4$ and $\delta = 1/(4k)$.
The choice of $\delta$ ensures that each subset of the returned list has size at most $k$.
To apply the theorem, we need $\omega \cdot (n/k^2)^t \leq \delta k$, which we rewrite as
\[k \geq \operatorname{poly}(t) \cdot \sqrt{n} \cdot n^{\frac{3}{8t+4}}\mper\]
Theorem~\ref{thm:rounding-from-biclique-certs} yields a list of $O((n/k)^{4t})$ subsets, each of size at most $k$, such that with probability at least $0.99$ the true clique $S^*$ is in the list.

Next, we remove from the list the subsets with size different than $k$ and the subsets that are not cliques.
Then we iterate the following procedure: find $S, S'$ in the list such that $|S \cap S'| \geq O(\log n)$ and remove one of them from the list.
By Lemma~\ref{lem:sets-pinex}, the resulting list has size at most $(1+o(1))n/k$, where we use that our choice of $k$ satisfies $k \geq \sqrt{2nO(\log n)}$.
Furthermore, by Lemma~\ref{lem:clique-intersection}, this procedure cannot remove $S^*$ from the list, because it intersects other cliques in at most $O(\log n)$ vertices.

We choose the smallest $t$ such that $\epsilon \geq \frac{3 + 0.1}{8t+4}$, which is $t=\lceil \frac{31}{80\epsilon}-\frac{1}{2}\rceil=O(1/\epsilon)$. 
Then $k \geq n^{1/2+\epsilon}$ satisfies the lower bounds on $k$ that we required in the proof.
Finally, the time complexity of the algorithm is $n^{O(t)} = n^{O(1/\epsilon)}$.
\end{proof}

\begin{lemma}[$n^{1/2+\epsilon}$ guarantee of Theorem~\ref{thm:main-result}, $p \geq 1/2$]\label{lem:result-p-large}
Fix any $\epsilon>0$ small enough. 
There is an algorithm that takes input a graph $G$ on $n$ vertices, runs in time $n^{O(1/\epsilon)}$, and provides the following guarantee:
If $G$ is generated according to $\FK(n,k,p)$ with $p \geq 1/2$, for $k \geq n^{1/2 + \epsilon}/(1-p)^{1/\epsilon}$, with probability $0.99$ the algorithm outputs a list of at most $(1+o(1))n/k$ $k$-cliques such that one of them is the planted clique in $G$.
\end{lemma}
\begin{proof}
First, we note that $\mathcal{A}(G)$ implies $\mathcal{A}(G')$ for any $G'$ that is obtained by adding edges to $G$. Therefore, in our sum-of-squares programs we can ignore the adversarial deletion phase and assume that we work with a graph in which the edges going out from $S^*$ are random.

By the second certificate in Theorem~\ref{thm:certification-large-p}, for $k \geq O(\sqrt{t n \log n} p^{2t} / (1-p)^{2t+1})$, we have
\begin{equation*}
\cB(H) \sststile{4t+2}{x,y} \Set{ |x|^{4t}|y| \leq 
(1000t)^{10t} n \Paren{\frac{np^{2t}}{k(1-p)^{2t+1}}}^4 }\mper 
\end{equation*}
Next, we want to apply Theorem~\ref{thm:rounding-from-biclique-certs} with $\omega=(1000t)^{10t} n \Paren{\frac{np^{2t}}{k(1-p)^{2t+1}}}^4$ and $\delta = 1/(4k)$. The choice of $\delta$ ensures that each subset of the returned list has size at most $k$.
To apply the theorem, we need $\omega \cdot (n/k^2)^t \leq \delta k$, which we rewrite as
\begin{align*}
k \geq \operatorname{poly}(t) \cdot \sqrt{n} \cdot n^{\frac{3}{8t+4}} p^{1-\frac{4}{8t+4}} / (1-p) \mper
\end{align*}
Theorem~\ref{thm:rounding-from-biclique-certs} yields a list of $O((n/k)^{4t})$ subsets, each of size at most $k$, such that with probability at least $0.99$ the true clique $S^*$ is in the list.

Next, we remove from the list the subsets with size different than $k$ and the subsets that are not cliques.
Then we iterate the following procedure: find $S, S'$ in the list such that $|S \cap S'| \geq O(\log n/(1-p))$ and remove one of them from the list.
By Lemma~\ref{lem:sets-pinex}, the resulting list has size at most $(1+o(1))n/k$, where we use that our choice of $k$ satisfies $k \geq \sqrt{2nO(\log n/(1-p))}$.
Furthermore, by Lemma~\ref{lem:clique-intersection}, this procedure cannot remove $S^*$ from the list, because it intersects other cliques in at most $O(\log n / \log 1/p) = O(\log n/(1-p))$ vertices, where we used that $\log 1/p = \Omega(1-p)$ for $p \geq 1/2$.

We choose the smallest $t$ such that $\epsilon \geq \frac{3 + 0.1}{8t+4}$, which is $t=\lceil \frac{31}{80\epsilon}-\frac{1}{2}\rceil=O(1/\epsilon)$. 
For this choice of $t$, we actually have $(1-p)^{2t+1} \geq (1-p)^{1/\epsilon}$ for $\epsilon \leq 0.1$. 
Then $k \geq n^{1/2+\epsilon}/(1-p)^{1/\epsilon}$ satisfies the lower bounds on $k$ that we require.
Finally, the time complexity of the algorithm is $n^{O(t)} = n^{O(1/\epsilon)}$.
\end{proof}

\begin{proof}[Proof of second guarantee in Theorem~\ref{thm:main-result}]
By Lemma~\ref{lem:result-p-small} we obtain the desired result for $p \leq 1/2$ when $k \geq n^{1/2+\epsilon}$, and by Lemma~\ref{lem:result-p-large} we obtain the desired result for $p \geq 1/2$ when $k \geq n^{1/2+\epsilon}/(1-p)^{1/\epsilon}$.
Then both results hold when $k \geq n^{1/2+\epsilon}/(1-p)^{1/\epsilon}$.
\end{proof}

\section{Evidence of hardness for certifying blicliques}
In this section, we collect some evidence that suggests that improving on our guarantees for the unbalanced bipartite clique certification problem is hard.
Our hardness results are in two settings: in the first we will prove a lower bound on the basic SDP relaxation for the problem of finding large bicliques in random graphs that gives a concrete reason for the $n^{2/3}$ barrier (for $p=1/2$) in prior works, and in the second we will prove lower bounds in the low-degree polynomial model for hypothesis testing problems. 

\subsection{Lower bounds against basic SDP}

We consider the following SDP relaxation for finding large bicliques in a given bipartite graph $H=(U, V, E)$ where $|U|=k$ and $|V|=n$.
It is equivalent to the degree $2$ sum-of-squares relaxation of the biclique constraint system \eqref{eq:biclique-system}.

\begin{equation} \label{eq:SDP-biclique}
  \left \{
    \begin{aligned}
      &\forall i,j 
      & 0 \leq X(i,j)
      & \leq 1 \\
      &
      &\tr(X) 
      &= k\\
      &
      &\Paren{\sum_{u \in U} X(u,u)} 
      &= \ell\\
      &
      &\Paren{\sum_{v \in V} X(v,v)} 
      &= k-\ell\\
      &
      &\sum_{u \in U,v \in V} X(u,v)
      &= \ell (k-\ell)\\
      &\forall u \in U,v \in V \text{ s.t. } \{u,v\} \not\in E
      & X(u,v)
      & = 0\\
      &
      & X
      & \succeq 0
    \end{aligned}
  \right \}
\end{equation}  

\paragraph{Commentary on the SDP relaxation.} We think of the SDP solution $X$ as a matrix indexed by all the (left and right) vertices of the bipartite graph $H$.
We associate the first $k=|U|$ rows and columns of $X$ with the left vertices and the last $n=|V|$ rows and columns of $X$ with the right vertices of $H$.
If $x \in \{0,1\}^{|U|}$ and $y \in \{0,1\}^{|V|}$ indicate left and right subsets of vertices in a purported biclique of total size $k$, then $X$ should be ``thought of'' as a relaxation of the constraints satisfied by the rank $1$ matrix $(x,y)(x,y)^{\top}$.
In particular, the first two constraints posit that $X$ is non-negative in all its entries and that its trace (that equals $\Norm{x}_2^2 + \Norm{y}_2^2$ and thus the total size of the biclique) is $k$.
The next three constraints posit that the left hand side vertices contribute $\ell$ to the trace (corresponding to the left hand side contributing $\ell$ vertices to the biclique), that the right hand side vertices contribute $k-\ell$ to the trace, and that $(\sum_i x_i) (\sum_i y_i) = \ell(k-\ell)$. 
The penultimate constraint posits that if $u\in U$ and $v \in V$ do not have an edge between them, then we cannot simultanesouly pick $u,v$ to be in the biclique (capturing the ``biclique'' constraints). 

For fixed $k,n$, the infeasibility of the SDP for some $\ell = \ell(k,n)$ is equivalent to there being a degree 2 sum-of-squares certificate of the absence of $\ell \times k-\ell$ bicliques in $H$.
We will show that the above SDP is in fact \emph{feasible} whp over the draw of $H$ so long as $\ell \ll n/k$.
This corresponds to the basic SDP barrier at $k=n^{2/3}$ (a threshold obtained by balancing the above obtained trade-off -- see Remark~\ref{remark:heuristic}) encountered in prior works on the semi-random planted clique problem.

\begin{lemma}[SDP lower bound for biclique certification] \label{lem:sdp-lb}
Let $H = (U, V, E) \sim B(k,n,1/2)$ be a bipartite \Erdos-\Renyi random graph with edge probability $1/2$. Then, with probability at least $0.99$ over the draw of $H$, for any $100 \sqrt{n} \leq k\leq n/2$ and $\ell \leq c n/k$ for some constant $c>0$ small enough, the SDP \eqref{eq:SDP-biclique} is feasible. 
\end{lemma}

\begin{proof}
We will prove the lemma by exhibiting an explicit solution to the SDP~\eqref{eq:SDP-biclique}.
The unbalanced setting requires a slightly more involved construction compared to the related SDP lower bounds for the clique number of $G(n,1/2)$, where a natural shifted and scaled adjacency matrix yields a feasible solution.

\paragraph{The construction.} In order to describe our construction, it is helpful to think of the solution $X$ as being divided into $X_{top}$, the principal $k \times k$ block corresponding to first $k$ rows and columns, $X_{bot}$, the principal $n \times n$ block corresponding to the last $n$ rows and columns, and $X_{cross}$, the $k \times n$ off-diagonal block (and its transposed copy). 

We will set every diagonal entry of $X_{top}$ to be $\ell/k$ and every off-diagonal entry of $X_{top}$ to be $(\ell/k)^2$.
Informally, $X_{top}$ is the 2nd moment matrix of the probability distribution that chooses every vertex on the left with probability $\ell/k$ independently. 

We describe $X_{cross}$ next.
For every $u \in U, v\in V$, we set $X(u,v)=X(v,u)=0$ if $u$ is not connected to $v$ in $H$, and otherwise we set $X(u,v) = X(v,u)=c_1 (\ell/n)$ for some constant $c_1>0$ to be chosen later.
Notice that this is equivalent to setting $X_{cross} = c_1 \frac{\ell}{n} A$ where $A$ is the $k$ by $n$ bipartite adjacency matrix of $H$. 

Finally, we describe $X_{bot}$.
This is where we need to be a little more careful.
Let $a_1, a_2,\ldots, a_k$ be $n$-dimensional vectors in $\{-1,1\}^n$ such that $a_u(v)=1$ iff $\{u,v\}$ is an edge in $H$.
That is, the $a_i$s are the $\pm 1$ neighborhood indicators of the $k$ left vertices in $H$.
Then, we set $X_{bot} = \frac{k-\ell}{n (k+1)} (\sum_{i=1}^k a_i a_i^{\top} + \1 \1^{\top})$.
Here $\1$ is the vector of all $1$ coordinates.
Note that the diagonal entries of $X_{bot}$ exactly equal $(k-\ell)/n$ and thus $\tr(X_{bot}) = k-\ell$ as required.
Also note that $X_{bot}$ is low rank, as it has rank at most $k+1$. 

We discuss now the choice of $c_1$. We want to enforce $\sum_{u \in U,v \in V} X(u, v) = \ell(k-\ell)$, so we choose $c_1 = \frac{\ell(k-\ell)}{(\ell/n)\sum_{u\in U, v \in V} A(u, v)} = \frac{(k-\ell)n}{\sum_{u \in U, v \in V} A(u, v)}$. Note that with probability at least $0.999$ we have that $\Omega(kn) \leq \sum_{u \in U, v \in V} A(u, v) \leq kn$, so with probability at least $0.999$ we have that $c_1$ is bounded below and above by absolute constants.

\paragraph{Analysis.} With probability at least $0.999$ over the draw of $H$, $X$ immediately satisfies  all the constraints except for positive semidefiniteness. We focus next on verifying the PSD-ness of $X$.
Consider any ``test'' vector $z \in \R^{k+n}$, which we will think of as $(z_L, z_R)$ where $z_L$ is the projection of $z$ to the first $k$ coordinates (i.e., the left vertices) and $z_R$ the projection to the last $n$ coordinates (i.e., the right vertices). 

Now, 
\begin{equation}
z^{\top} X z = z_L^{\top} X_{top} z_L + z_R^{\top} X_{bot} z_R + 2 z_L^{\top} X_{cross} z_R \mper \label{eq:quad-form}
\end{equation} 

Let $F$ be the subspace of at most $k+1$ dimensions spanned by the $k$ rows of $A$ and the all $1$s vector $\1$.
Now, notice that $X_{cross} z_R = X_{cross} z_R^F$ where $z_R^F$ is the projection of $z_R$ to $F$.
Similarly, by design, $X_{bot}$ has range space equal to $F$, so $X_{bot} z_R = X_{bot} z_R^F$.
Thus, WLOG, we can assume that $z_R = z_R^F$ in the following. 

Let's write $z_L = z_L^{\parallel} + z_L'$ and $z_R = z_R^{\parallel} + z_R'$ where $z_L^{\parallel} = \iprod{z_L, \frac{\1}{\sqrt{n}}} \frac{\1}{\sqrt{n}}$ is the component of $z_L$ along the all $1$s direction (and similarly for $z_R^{\parallel}$). Let $X_{cross}' = c_1 \frac{\ell}{n} (A-\frac{1}{2} \1 \1^{\top})$ be the ``centered'' version of $X_{cross}$.
Also define the centered versions $X_{top}' = X_{top} - \frac{\ell^2}{k^2} \1 \1^{\top}$ and $X_{bot}' = \frac{k-\ell}{n(k+1)} \sum_{i =1}^k a_i a_i^{\top}$. 

Our argument is to simply ``charge'' the third term (which can be potentially negative) to the first two terms (that are always non-negative).
We will use the following two standard random matrix facts (see Fact~\ref{fact:vershynin}) in our analysis: for $A' = A- \frac{1}{2} \1 \1^{\top}$ we have $\Norm{A'}_2 \leq O(\sqrt{n})$, and the $k$-th smallest singular value of both $A$ and $A'$ is at least $\Omega(\sqrt{n} - \sqrt{k}) = \Omega(\sqrt{n})$ as $k \leq n/2$. 

\paragraph{The potentially negative terms.} 
Let's work with the potentially negative terms coming from the parallel components of $z_L$ and $z_R$. 
Observe that $(z_L^{\parallel})^{\top} X_{cross} z_R^{\parallel} = \Norm{z_L^{\parallel}}_2 \Norm{z_R^{\parallel}}_2 \frac{c_1}{2} \ell + (z_L^{\parallel})^{\top} X_{cross}' z_R^{\parallel} \geq \Norm{z_L^{\parallel}}_2 \Norm{z_R^{\parallel}}_2 (\frac{c_1 \ell}{2} - O(\frac{\ell}{n}\sqrt{n})) > 0$.
Since $(z_L^{\parallel})^{\top} X_{top} z_L^{\parallel} +(z_R^{\parallel})^{\top} X_{bot} z_R^{\parallel} >0$ we can conclude that $(z^{\parallel})^{\top} X z^{\parallel}\geq 0$.

Let's analyze the potentially negative terms coming from the perpendicular components of $z_L$ and $z_R$. 
We have $|(z_L')^{\top} X_{cross} z_R'| = |(z_L')^{\top} X_{cross}' z_R'| \leq c_2 \frac{\ell}{n} \sqrt{n} \norm{z_L'}_2 \norm{z_R'}_2 = c_2 \frac{\ell}{\sqrt{n}} \norm{z_L'}_2 \norm{z_R'}_2$. 

Finally, let's analyze the potentially negative terms coming from crossing the parallel and the perpendicular components of $z_L$ and $z_R$. 
We have: $|(z_L^{\parallel})^\top X_{cross} z_R'| = |(z_L^{\parallel})^\top X_{cross}' z_R'| \leq \norm{z_L^{\parallel}}_2 \norm{z_R'}_2 O(\ell/\sqrt{n})$.
Similarly, $|(z_L')^\top X_{cross} z_R^{\parallel}| \leq  \norm{z_L'}_2 \norm{z_R^{\parallel}}_2 O(\ell/\sqrt{n})$.

\paragraph{The square terms.}
We now compute a lower bound on the non-negative terms in \eqref{eq:quad-form}.

We have $z_L^\top X_{top} z_L = z_L^\top X_{top}' z_L + \frac{\ell^2}{k^2} (z_L^{\parallel})^\top \1 \1^\top z^{\parallel} = \norm{z_L}_2^2 (\frac{\ell}{k} - \frac{\ell^2}{k^2}) + \norm{z_L^{\parallel}}_2^2 \frac{\ell^2 n}{k^2} \geq c_3 \norm{z_L}_2^2 \frac{\ell}{k}$.

Next, we lower bound $z_R^{\top} X_{bot} z_R$. Now, $z_R\in F$. Recall that the $k$-th smallest singular value of $A$ and $A'$ is at least $\Omega(\sqrt{n})$ if $k \leq n/2$ with probability at least $0.999$ over the draw of the graph $H$, and 2) the matrix $\sum_i a_i a_i^{\top}$ has the same eigenvalues as $4A'{A'}^{\top}$ where, recall that $A' = A-\frac{1}{2} \1 \1^{\top}$. 
Together this yields that all eigenvalues of $\sum_i a_i a_i^{\top} + \1 \1^{\top}$ are at least $c_4 n$ for some constant $c_4>0$ when restricted to the subspace $F$. 
Thus, $z_R^{\top} X_{bot} z_R \geq c_4 n \Norm{z_R}_2^2 \frac{k-\ell}{n (k+1)} = c_4 \Norm{z_R}_2^2 \frac{k-\ell}{k+1} \geq c_5 \Norm{z_R}_2^2$ recalling that $k-\ell > k/2$. 

Let's now complete the charging argument. 
Let's first observe that, by the AM-GM inequality, the square terms contribute at least $c_6 \sqrt{\ell/k}\norm{z_L}_2 \norm{z_R}_2$. The potentially negative term from the perpendicular components is at most $c_2 \ell/\sqrt{n} \norm{z_L'}_2 \norm{z_R'}_2$ in magnitude, and the potentially negative term from crossing the components is at most $\norm{z_L^{\parallel}}_2 \norm{z_R'}_2 O(\ell/\sqrt{n}) + \norm{z_L'}_2 \norm{z_R^{\parallel}}_2 O(\ell/\sqrt{n})$.
Thus, the square terms dominate as long as $\ell \leq O(n/k)$. 

This completes the proof.

\end{proof}

\subsection{Low-degree lower bound for $p=1/2$}
Formally, we will prove that there are distributions over bipartite graphs that admit $\ell$ by $k-\ell$ cliques for appropriate parameters $\ell$ that are indistinguishable from $B(k,n,p)$ --- the distribution on random bipartite graphs with left vertex set of size $k$, right vertex set of size $n$ and each bipartite edge included to be in the graph with probability $p$ independently.
The choice of the planted model requires a bit of care, as we soon discuss. We will deal with the case of $p=1/2$ and general $p$ separately for clarity of exposition.

\begin{itemize}
  \item $D_{\mathrm{null}} = B(k,n, 1/2)$: the distribution on bipartite graphs $H=(U,V, E)$ where $|U|=k$, $|V|=n$ and each bipartite edge  $\set{u,v}$ with $u \in U$ and $v \in V$ is included in $H$ with probability $1/2$.
  \item $D_{\mathrm{planted}} = B(k,n, \ell, 1/2)$: the distribution on bipartite graphs $H=(U,V, E)$ where $|U|=k$, $|V|=n$, sampled as follows. 
  Choose $S$ by including each vertex from $U$ in $S$ with probability $\ell/k$.
  Choose $P$ by including every vertex from $V$ in $R$ with probability $(k-\ell)/n$.
  Finally, include each edge $\set{u,v}$ with $u \in U$ and $v \in V$ with probability
\[\Pr_{D_{\mathrm{planted}}}[\set{u, v} \text{ is included}] = \begin{cases}
1 & \text{ if } u \in S, v \in P\mcom\\
\frac{n/2-(k-\ell)}{n-(k-\ell)} & \text{ if } u \in S, v \not\in P\mcom\\
\frac{1}{2} & \text{ otherwise}\mper
\end{cases}\]
\end{itemize}

\begin{remark}
$D_{\mathrm{planted}}$ is chosen so as to have a $\ell$ by $k-\ell$ biclique in it while having the same distribution of degrees of \emph{left} vertices as in $D_{\mathrm{null}}$.
This is necessary since otherwise the average degree of the left vertices gives a distinguisher between the models.
\end{remark}

\begin{theorem}\label{thm:low-deg-lb-half}
Fix $\epsilon > 0$ independent of $n$ with $\epsilon \leq 0.001$.
For $k=n^{1/2+\epsilon}$ and $\ell \leq n^{1/4-0.001}$, the norm of the degree-$\lfloor 0.001/\epsilon \rfloor$  likelihood ratio between $B(k,n,\ell, 1/2)$ and $B(k,n, 1/2)$ is $1+o(1)$.
On the other hand, for $k=n^{1/2+\epsilon}$ and all $\ell$, the norm of the degree-$O(1/\epsilon)$ likelihood ratio between $B(k,n,\ell, 1/2)$ and $B(k,n, 1/2)$ is unbounded as $n \to \infty$.
\end{theorem}

\begin{remark}
Information-theoretically, to identify a small list in the semi-random planted clique model, we need $k = \tilde{\Theta}(\sqrt{n})$~\cite{steinhardt2017does}.
If we set $k$ to be this value, then the corresponding bipartite random graph has no $\ell$ by $k-\ell$ clique for $\ell = O(\log n)$. 
The above theorem shows that in the low-degree polynomial model, distinguishing between the case when $\ell = n^{\epsilon}$ vs $\ell = O(\log n)$ requires polynomials of degree $O(1/\epsilon)$.
\end{remark}

For a bipartite graph $H=(U,V, E)$ recall that $\chi_{u,v}$ is $1$ if the edge $\set{u, v}$ is included and $-1$ otherwise.
We also define $\chi_\alpha = \prod_{\set{u, v} \in A} \chi_{u,v}$.

\begin{lemma}
\label{lem:low-deg-lb-mono-half}
For $H$ sampled from $D_{\mathrm{planted}}$, let $L$ be the number of left vertices in $\alpha$, $R$ the number of right vertices in $A$, and $d_1, ..., d_R$ the number of edges in $\alpha$ incident to each of the right vertices.
Then
\[\E_{D_{\mathrm{planted}}}[\chi_\alpha] = \begin{cases}
\Paren{\frac{\ell}{k}}^L \Paren{\frac{k-\ell}{n}}^R \Paren{1 + O\Paren{\frac{k-\ell}{n}}}^R & \text{ if } d_1, \ldots, d_R > 1\mcom\\
0 & \text{ otherwise}\mper
\end{cases}\]
\end{lemma}
\begin{proof}
Conditioned on the planted biclique $(S, P)$, the edges are independent. For an edge $\set{u, v}$, we calculate
\[\E_{D_{\mathrm{planted}}}[\chi_{u,v} \mid \text{planted biclique is }(S,P)] = \begin{cases}
  1 & \text{ if } u \in S, v \in P\mcom\\
  \frac{-(k-\ell)}{n-(k-\ell)} & \text{ if } u \in S, v \not\in P \mcom\\
  0 & \text{ otherwise}\mcom
\end{cases}\]
where for the case $u \in S, v \not\in P$, we calculated the expectation as
\[\frac{n/2-(k-\ell)}{n-(k-\ell)} \cdot 1 + \Paren{1 - \frac{n/2-(k-\ell)}{n-(k-\ell)}} \cdot (-1) = \frac{-(k-\ell)}{n-(k-\ell)}\mper\]

We observe that if any of the left vertices in $\alpha$ is not in the planted biclique, the conditional expectation of $\chi_\alpha$ is zero.
Therefore, we condition on the event that all the left vertices in $A$ are in the planted biclique, which happens with probability $\Paren{\frac{\ell}{k}}^L$.

Conditioned on this event, for any particular right vertex, all the edges in $\alpha$ incident to it are independent from the other edges in $\alpha$.
Let $\alpha_i$ be the subset of edges in $\alpha$ that are incident to the $i$-th right vertex.
Then 
\begin{align*}
&\E_{D_{\mathrm{planted}}}[\chi_{\alpha_i} \mid \text{ planted biclique contains all left vertices in } \alpha]\\
&= \frac{k-\ell}{n} \cdot 1 + \Paren{1 - \frac{k-\ell}{n}} \cdot \Paren{\frac{-(k-\ell)}{n-(k-\ell)}}^{d_i}\\
&= \begin{cases}
\frac{k-\ell}{n} \Paren{1 + O\Paren{\frac{k-\ell}{n}}} & \text{ if } d_i > 1\mcom\\
0 & \text{ if } d_i = 1\mcom
\end{cases}
\end{align*}
so
\begin{align*}
&\E_{D_{\mathrm{planted}}}[\chi_{\alpha} \mid \text{ planted biclique contains all left vertices in } \alpha]\\
&= \begin{cases}
\Paren{\frac{k-\ell}{n}}^R \Paren{1 + O\Paren{\frac{k-\ell}{n}}}^R & \text{ if } d_1, \ldots, d_R > 1\mcom\\
0 & \text{ otherwise }\mper 
\end{cases}
\end{align*}
Therefore, overall, 
\[\E_{D_{\mathrm{planted}}}[\chi_\alpha] = \begin{cases}
\Paren{\frac{\ell}{k}}^L \Paren{\frac{k-\ell}{n}}^R \Paren{1 + O\Paren{\frac{k-\ell}{n}}}^R & \text{ if } d_1, \ldots, d_R > 1\mcom\\
0 & \text{ otherwise }\mper 
\end{cases}\]
\end{proof}

\begin{proof}[Proof of Theorem~\ref{thm:low-deg-lb-half}]
Let $LR^{\leq D}$ be the degree-$D$ likelihood ratio between $D_{\mathrm{planted}}$ and $D_{\mathrm{null}}$.
Then, by standard results, $\Norm{LR^{\leq D} - 1}^2 = \sum_{0 < |A| \leq D} \E_{D_{\mathrm{planted}}}[\chi_\alpha]^2$, where the norm is the one induced by $D_{\mathrm{null}}$.
Therefore, if the right-hand side is $o(1)$, then $\Norm{LR^{\leq D}}$ is $1+o(1)$, and if the right-hand side is unbounded, then $\Norm{LR^{\leq D}}$ is also unbounded.

Consider all $E$ with $L$ left vertices and $R$ right vertices.
The contribution from these $E$ is, by Lemma~\ref{lem:low-deg-lb-mono-half}, 
\[\sum_{\substack{\alpha\\L \text{ left vertices}\\R \text{ right vertices}}} \E_{D_{\mathrm{planted}}}[\chi_\alpha]^2 = \Paren{\frac{\ell}{k}}^{2L} \Paren{\frac{k-\ell}{n}}^{2R} \Paren{1 + O\Paren{\frac{k-\ell}{n}}}^{2R} \cdot \binom{k}{L} \binom{n}{R} \operatorname{Bip}(L, R)\mcom\]
where $\operatorname{Bip}(L,R)$ is the number of bipartite graphs with $L$ left vertices and $R$ right vertices such that all left degrees are at least $1$ and all right degrees are greater than $1$.

Consider a choice of $L$ and $R$ such that $\operatorname{Bip}(L,R) \neq 0$.
Because we are interested in the behavior of the sum as $n$ goes to infinity, we ignore as negligible all factors that depend only on $L$ and $R$.
We also approximate $k-\ell \approx k$ and $\Paren{1+O\Paren{\frac{k-\ell}{n}}}^{2R} \approx 1$.
Then we have 
\begin{align*}
\sum_{\substack{\alpha\\L \text{ left vertices}\\R \text{ right vertices}}} \E_{D_{\mathrm{planted}}}[\chi_\alpha]^2
&\sim \Paren{\frac{\ell}{k}}^{2L} \Paren{\frac{k}{n}}^{2R} \cdot \binom{k}{L} \binom{n}{R}\\
&\sim \Paren{\frac{\ell}{k}}^{2L} \Paren{\frac{k}{n}}^{2R} k^L n^R\\
&= n^{-R}k^{2R-L}\ell^{2L}\mper
\end{align*}
For $k = n^{1/2+\epsilon}$, the above is equal to $n^{(2R-L)\epsilon - L/2}\ell^{2L}$.

For $\ell=n^{1/4-0.001}$, this is equal to $n^{(2R-L)\epsilon - 0.002L}$. 
For $|\alpha| \leq 0.001/\epsilon$, we have $1 \leq L,R \leq 0.001/\epsilon$ and hence $2R-L \leq 0.002/\epsilon-1$.
Then $n^{(2R-L)\epsilon - 0.002L} \leq n^{-\epsilon}$, which goes to zero as $n$ goes to infinity.
Therefore, the sum of all the terms with $|\alpha| \leq 0.001/\epsilon$ is $o(1)$.

For $|\alpha| \geq O(1/\epsilon)$, consider the term corresponding to $L=2$ and some $R=O(1/\epsilon)$.
Note that the term satisfies $\operatorname{Bip}(L, R) \neq 0$ (e.g., the complete bipartite graph on $2$ left vertices and $O(1/\epsilon)$ right vertices is a valid choice).
For this term, $n^{(2R-L)\epsilon - L/2} = n^{2R\epsilon - 2\epsilon - 1} \geq n$ for $R = O(1/\epsilon)$ large enough.
Therefore, this term goes to infinity as $n$ goes to infinity, and then the same is true for the sum of all the terms.
\end{proof}

\subsection{Low-degree lower bound for general densities}

In this section, we will prove that the following two distributions on bipartite random graphs are indistinguishable by low-degree polynomials. 

\begin{itemize}
  \item $D_{\mathrm{null}} = B(k, n, p)$: the distribution on bipartite graphs $H=(U,V, E)$ where $|U|=k$, $|V|=n$, and each bipartite edge  $\set{u, v}$ with $u \in U$ and $v \in V$ is included in $H$ with probability $p$.
  \item $D_{\mathrm{planted}} = B(k,n, \ell, p)$: the distribution on bipartite graphs $H=(U,V, E)$ where $|U|=k$, $|V|=n$, sampled as follows. 
  Choose $S$ by including each vertex from $U$ in $S$ with probability $\ell/k$.
  Choose $P$ by including every vertex from $V$ in $R$ with probability $(k-\ell)/n$.
  Finally, include each edge $\set{u, v}$ with $u \in U$ and $v \in V$ with probability
  \[\Pr_{D_{\mathrm{planted}}}[\set{u, v} \text{ is included}] = \begin{cases}
  1 & \text{ if } u \in S, v \in P\mcom\\
  \frac{np-(k-\ell)}{n-(k-\ell)} & \text{ if } u \in S, v \not\in P\mcom\\
  p & \text{ otherwise}\mper
\end{cases}\]
\end{itemize}

\begin{theorem}
\label{thm:low-deg-lb-general}
Fix $\epsilon > 0$ independent of $n$.
Let $p \geq 1/2$ and $q=1-p$, and define $\gamma$ such that $q=n^{-\gamma}$.
For $\epsilon \leq \gamma/2$ and $k=n^{1/2+\epsilon}/q^{1/2}$ and $\ell\leq n^{1/4 - 0.001}$, the norm of the degree-$f(1/\epsilon)$  likelihood ratio between $B(k,n,\ell, p)$ and $B(k,n, p)$ is $1+o(1)$, for any function $f$ independent of $n$.
On the other hand, for $\epsilon \geq \gamma$ and $k=n^{1/2+\epsilon}/q^{1/2}$ and all $\ell$, the norm of the degree-$O(1/\epsilon)$  likelihood ratio between $B(k,n,\ell, p)$ and $B(k,n, p)$ is unbounded as $n \to \infty$.
\end{theorem}
\begin{remark}
Information-theoretically, to identify a small list in the semi-random planted clique model with a general $p$, we need $k \sim \tilde{\Theta}(\sqrt{n/q})$~\cite{steinhardt2017does}.
If we set $k$ to be this allegedly optimal value, then the corresponding bipartite random graph has no $\ell$ by $k-\ell$ clique for $\ell = O(\log n)$. 
The above theorem shows that in the low-degree polynomial model, distinguishing between the case when $\ell = n^{\epsilon}$ vs $\ell = O(\log n)$ requires polynomials of degree growing faster than any function (independent of $n$) of $1/\epsilon$.\footnote{The theorem leaves open the possibility of distinguishing with low degree in the case $\epsilon > \gamma/2$. However, if $\epsilon > \gamma/2$, then $k \geq \sqrt{n}/q$, which is also suboptimal.} 
\end{remark}

In this section, for a bipartite graph $H=(U,V, E)$, we define $\chi_{u,v}$ to be $\sqrt{\frac{1-p}{p}}$ if the edge $\set{u, v}$ is included and $-\sqrt{\frac{p}{1-p}}$ otherwise.
We also define $\chi_\alpha = \prod_{\set{u, v} \in E} \chi_{u,v}$.

\begin{lemma}
\label{lem:low-deg-lb-mono-general}
For $H$ sampled from $D_{\mathrm{planted}}$, let $L$ be the number of left vertices in $\alpha$, $R$ the number of right vertices in $E$, and $d_1, ..., d_R$ the number of edges in $\alpha$ incident to each of the right vertices.
Then
\[\E_{D_{\mathrm{planted}}}[\chi_\alpha] = \begin{cases}
  \Paren{\frac{\ell}{k}}^L \Paren{\frac{k-\ell}{n}}^R \prod_{i=1}^R \Paren{\Paren{\sqrt{\frac{1-p}{p}}}^{d_i}+O\Paren{\frac{k-\ell}{n}\frac{1-p}{p}}} & \text{ if } d_1, \ldots, d_R > 1\mcom\\
  0 & \text{ otherwise }\mper 
  \end{cases}\]
\end{lemma}
\begin{proof}
Conditioned on the planted biclique $(S, P)$, the edges are independent. 
For an edge $\set{u, v}$, we calculate
\[\E_{D_{\mathrm{planted}}}[\chi_{u,v} \mid \text{planted biclique is }(S,P)] = \begin{cases}
  \sqrt{\frac{1-p}{p}} & \text{ if } u \in S, v \in P\mcom\\
  \frac{-\frac{k}{n}}{1-\frac{k}{n}} \sqrt{\frac{1-p}{p}} & \text{ if } u \in S, v \not\in P \mcom\\
  0 & \text{ otherwise}\mcom
\end{cases}\]
where for the case $u \in S, v \not\in P$, we calculated the expectation as
\[\frac{np-(k-\ell)}{n-(k-\ell)} \cdot \sqrt{\frac{1-p}{p}} + \Paren{1-\frac{np-(k-\ell)}{n-(k-\ell)}} \cdot \Paren{-\sqrt{\frac{p}{1-p}}} = \frac{-(k-\ell)}{n-(k-\ell)} \sqrt{\frac{1-p}{p}}\mper\]

We observe that if any of the left vertices in $\alpha$ is not in the planted biclique, the conditional expectation of $\chi_\alpha$ is zero.
Therefore, we condition on the event that all the left vertices in $\alpha$ are in the planted biclique, which happens with probability $\Paren{\frac{\ell}{k}}^L$.

Conditioned on this event, for any particular right vertex, all the edges in $\alpha$ incident to it are independent from the other edges in $\alpha$.
Let $\alpha_i$ be the subset of edges in $\alpha$ that are incident to the $i$-th right vertex.
Then 
\begin{align*}
&\E_{D_{\mathrm{planted}}}[\chi_{\alpha_i} \mid \text{ planted biclique contains all left vertices in } \alpha]\\
&= \frac{k-\ell}{n} \cdot \Paren{\sqrt{\frac{1-p}{p}}}^{d_i} + \Paren{1 - \frac{k-\ell}{n}} \cdot \Paren{\frac{-\frac{k-\ell}{n}}{1-\frac{k-\ell}{n}}\sqrt{\frac{1-p}{p}}}^{d_i}\\
&= \begin{cases}
\frac{k-\ell}{n} \Paren{ \Paren{\sqrt{\frac{1-p}{p}}}^{d_i} + O\Paren{\frac{k-\ell}{n} \frac{1-p}{p}}} & \text{ if } d_i > 1\mcom\\
0 & \text{ if } d_i = 1\mcom
\end{cases}
\end{align*}
so
\begin{align*}
&\E_{D_{\mathrm{planted}}}[\chi_{\alpha} \mid \text{ planted biclique contains all left vertices in } \alpha]\\
&= \begin{cases}
\Paren{\frac{k-\ell}{n}}^R \prod_{i=1}^R \Paren{\Paren{\sqrt{\frac{1-p}{p}}}^{d_i}+O\Paren{\frac{k-\ell}{n}\frac{1-p}{p}}} & \text{ if } d_1, \ldots, d_R > 1\mcom\\
0 & \text{ otherwise }\mper 
\end{cases}
\end{align*}
Therefore, overall, 
\[\E_{D_{\mathrm{planted}}}[\chi_\alpha] = \begin{cases}
\Paren{\frac{\ell}{k}}^L \Paren{\frac{k-\ell}{n}}^R \prod_{i=1}^R \Paren{\Paren{\sqrt{\frac{1-p}{p}}}^{d_i}+O\Paren{\frac{k-\ell}{n}\frac{1-p}{p}}} & \text{ if } d_1, \ldots, d_R > 1\mcom\\
0 & \text{ otherwise }\mper 
\end{cases}\]
\end{proof}

\begin{proof}[Proof of Theorem~\ref{thm:low-deg-lb-general}]
Let $LR^{\leq D}$ be the degree-$D$  likelihood ratio between $D_{\mathrm{planted}}$ and $D_{\mathrm{null}}$.
Then, by standard results, $\Norm{LR^{\leq D} - 1}^2 = \sum_{0 < |\alpha| \leq D} \E_{D_{\mathrm{planted}}}[\chi_\alpha]^2$, where the norm is the one induced by $D_{\mathrm{null}}$.
Therefore, if the right-hand side is $o(1)$, then $\Norm{LR^{\leq D}}$ is $1+o(1)$, and if the right-hand side is unbounded, then $\Norm{LR^{\leq D}}$ is also unbounded.

Consider all $E\alpha$ with $L$ left vertices and $R$ right vertices.
The contribution from these $\alpha$ is, by Lemma~\ref{lem:low-deg-lb-mono-general},
\begin{align*}
&\sum_{\substack{\alpha\\L \text{ left vertices}\\R \text{ right vertices}}} \E_{D_{\mathrm{planted}}}[\chi_\alpha]^2\\
&= \sum_{d_1, \ldots, d_R > 1} \Paren{\frac{\ell}{k-\ell}}^{2L} \Paren{\frac{k-\ell}{n}}^{2R} \prod_{i=1}^R \Paren{\Paren{\sqrt{\frac{1-p}{p}}}^{d_i}+O\Paren{\frac{k-\ell}{n}\frac{1-p}{p}}}^2 \cdot \binom{k}{L} \binom{n}{R} \operatorname{Bip}(L, R, d_1, \ldots, d_R)\mcom
\end{align*}
where $d_1, \ldots, d_R$ represent the number of edges in $\alpha$ incident to each of the right vertices, and $\operatorname{Bip}(L,R, d_1, \ldots, d_R)$ is the number of bipartite graphs with $L$ left vertices and $R$ right vertices such that all left degrees are at least $1$ and the right degrees are $d_1, \ldots, d_R$.

Consider a choice of $L$ and $R$ such that $\sum_{d_1, ..., d_R} \operatorname{Bip}(L, R, d_1, \ldots, d_R) \neq 0$.
Because we are interested in the behavior of the sum as $n$ goes to infinity, we ignore as negligible all factors that depend only on $L$ and $R$.
In particular, because $\sum_{d_1, ..., d_R} \operatorname{Bip}(L, R, d_1, \ldots, d_R)$ can be bounded in terms of only $L$ and $R$, we can focus on the term corresponding to $d_1=\ldots=d_R=2$, which maximizes the contribution of the terms $\sqrt{\frac{1-p}{p}}^{d_i}$ and is therefore proportional to the entire sum up to factors that depend only on $L$ and $R$.
We also approximate $k-\ell \approx k$ and $\Paren{\frac{1-p}{p}+O\Paren{\frac{k-\ell}{n}\frac{1-p}{p}}}^{2R} \approx \Paren{\frac{1-p}{p}}^{2R}$. Then we have 
\begin{align*}
\sum_{\substack{\alpha\\L \text{ left vertices}\\R \text{ right vertices}}} \E_{D_{\mathrm{planted}}}[\chi_\alpha]^2
&\sim \Paren{\frac{\ell}{k}}^{2L} \Paren{\frac{k}{n}}^{2R} \Paren{\frac{1-p}{p}}^{2R} \cdot \binom{k}{L} \binom{n}{R}\\
&\sim \Paren{\frac{\ell}{k}}^{2L} \Paren{\frac{k}{n}}^{2R} (1-p)^{2R} k^L n^R\\
&= n^{-R}k^{2R-L}\ell^{2L}(1-p)^{2R}\mper
\end{align*}
For $k = n^{1/2+\epsilon}/(1-p)^{1/2}$, the above is equal to $n^{(2R-L)\epsilon-L/2}\ell^{2L}(1-p)^{R+L/2}$.
For $1-p=q=n^{-\gamma}$, this is equal to $n^{(2R-L)\epsilon-L/2-(R+L/2)\gamma}\ell^{2L}$.
Finally, for $\ell=n^{\delta}$, this is equal to $n^{(2R-L)\epsilon+2L\delta-L/2-(R+L/2)\gamma}$.

For $\epsilon \leq \gamma/2$, we have that the above is at most $n^{-L\epsilon+2L\delta-L/2-L\gamma/2}$.
For the exponent to be non-negative, we need $2L\delta \geq L/2$, so $\delta \geq 1/4$.
In particular, for $\delta \leq 1/4-0.001$, the term goes to zero as $n$ goes to infinity regardless of how large $L$ is.
Therefore, the sum of all the terms with $|\alpha| \leq f(1/\epsilon)$ is $o(1)$, for any function $f$ independent of $n$.

For $\epsilon \geq \gamma$, consider the term corresponding to $L=2$ and some $R=O(1/\epsilon)$.
We have that the term is at least $n^{R\epsilon - L\epsilon + 2L\delta - L/2 - L\gamma/2} = n^{R\epsilon - 2\epsilon + 4\delta - 1 - \gamma} \geq n$ for $R=O(1/\epsilon)$ large enough.
Therefore, this term goes to infinity as $n$ goes to infinity, and then the same is true for the sum of all the terms.

\end{proof}

\bibliographystyle{alpha}
\bibliography{bib/custom,bib/dblp,bib/mathreview,bib/scholar,bib/references,bib/witmer}
\end{document}